%% file: main.tex
\documentclass{article}

\input{packages}
\input{preamble}

\newcommand{\kdm}{\textsf{$k$-DM}}
\newcommand{\mwbm}{\textsf{MW$b$M}}
\newcommand{\mwm}{\textsf{MWM}}
\newcommand{\mwkm}{\textsf{MW$k$M}}
\newcommand{\primal}{\hyperlink{primal-k-djm}{\textsf{(P)}}}
\newcommand{\dual}{\hyperlink{dual-k-djm}{\textsf{(D)}}}

\title{Semi-Streaming Algorithms for Weighted $k$-Disjoint Matchings\thanks{Source code: \url{https://github.com/smferdous1/GraST}}} 
\author{
    S M Ferdous\thanks{Pacific Northwest National Laboratory. Email:~\email{sm.ferdous@pnnl.gov}. Supported by the Laboratory Directed Research and Development Program at PNNL.}
    \and 
    Bhargav Samineni\thanks{The University of Texas at Austin. Email:~\email{sbharg@utexas.edu}. Supported by the U.S. DOE Science Undergraduate Laboratory Internships (SULI) program.}
    \and
    Alex Pothen\thanks{Purdue University. Email:~\email{apothen@purdue.edu}. Supported by U.S. Department of Energy SC-0022260.}
    \and
    Mahantesh Halappanavar\thanks{Pacific Northwest National Laboratory. Email:~\email{mahantesh.halappanavar@pnnl.gov}. Supported by U.S. Department of Energy 17-SC-20-SC (ECP ExaGraph) at PNNL.}
    \and
    Bala Krishnamoorthy\thanks{Washington State University Vancouver. Email:~\email{kbala@wsu.edu}.}
}

\begin{document}
    \maketitle

    \input{sections/abstract.tex}

    \input{sections/intro.tex}
    \input{sections/preliminaries.tex}
    \input{sections/related_work.tex}

    \input{sections/primal_dual.tex}
    \input{sections/bmatching.tex}

    \input{sections/heuristics.tex}
    \input{sections/experiment.tex}
    \input{sections/conclusion.tex}

    \newpage
    \bibliography{refs}

    \newpage
    \appendix 
    \input{sections/appendix_additional.tex}
    \input{sections/appendix_algorithms.tex}

\input{sections/appendix_experiments.tex}
\end{document}

%% file: packages.tex
\usepackage[letterpaper, margin=1in]{geometry}

\usepackage[utf8]{inputenc}
\usepackage{microtype} 
\usepackage{lmodern}
\usepackage[T1]{fontenc}
\usepackage[english]{babel}
\usepackage{comment}
\usepackage[inline]{enumitem}
\usepackage{booktabs}
\usepackage{multicol}
\usepackage{multirow}
\usepackage{relsize}
\usepackage{dblfloatfix}

\usepackage{amsmath}
\usepackage{amssymb}
\usepackage{mathtools}
\usepackage{amsthm}
\usepackage[c2, nocomma, short]{optidef}
\usepackage{bm}
\usepackage{nicefrac}

\usepackage{color}
\usepackage[dvipsnames]{xcolor}
\definecolor{linkcolor}{rgb}{0, 0.25, 0.75}

\usepackage[
  colorlinks=true,
  linkcolor=linkcolor,
  citecolor=linkcolor,
  urlcolor=Maroon,
  bookmarksnumbered,
  pagebackref,
]{hyperref}
\usepackage[nameinlink,capitalize]{cleveref}
\usepackage{thmtools}
\usepackage{thm-restate}
 
\usepackage[ruled]{algorithm}
\usepackage[endLComment=, commentColor=ForestGreen]{algpseudocodex}

\usepackage{tikz}
\usepackage{pgfplots}
\pgfplotsset{compat=1.17}

\usepackage{graphicx}
\usepackage{subcaption}
\usepackage{todonotes}
\usepackage{xspace}

%% file: preamble.tex
\def\eps{\varepsilon}

\newcommand{\CalS}{\mathcal{S}}
\newcommand{\CalM}{\mathcal{M}}

\newcommand{\mapprox}{$\frac{1}{3+\eps}$}

\newcommand{\lrp}[1]{\left( #1 \right)}
\newcommand{\lrb}[1]{\left[ #1 \right]}
\newcommand{\lrc}[1]{\left\{ #1 \right\}}
\newcommand{\lrv}[1]{\left\langle #1 \right\rangle}
\newcommand{\abs}[1]{\left\lvert #1 \right\rvert}

\newcommand{\Z}{\ensuremath{\mathbb{Z}}}
\newcommand{\R}{\ensuremath{\mathbb{R}}}

\DeclareMathOperator*{\argmin}{arg\,min}

\newcommand{\NPH}{NP-hard}
\newcommand{\APXH}{APX-hard}

\newcommand{\bigO}[2][]{\ensuremath{O_{#1}(#2)}}
\newcommand{\bigOlog}[2][]{\ensuremath{\tilde{O_{#1}}(#2)}}

\newcommand{\degree}[2][]{\mathrm{deg}_{#1}(#2)}

\newcommand{\poly}[1]{\ensuremath{\mathrm{poly}(#1)}}
\newcommand{\polylog}[1]{\ensuremath{\mathrm{polylog}(#1)}}

\newcommand{\qedsym}{\hfill$\blacksquare$}
\declaretheoremstyle[
  spacebelow=6pt,
  headfont=\itshape,
  bodyfont=\normalfont,
  notefont=\itshape, notebraces={}{},
  qed=\qedsym
]{proof}

\declaretheorem[style=plain]{theorem, lemma, corollary, claim, observation, proposition, result, question}
\declaretheorem[style=plain]{definition, conjecture, remark, problem}

\declaretheorem[numbered=no, style=proof]{proof}

\def\sse{\subseteq}
\def\sm{\setminus}

\def\calM{\mathcal{M}}

\def\calP{\mathcal{P}}

\def\calC{\mathcal{C}}

\def\calS{\mathcal{S}}

\newcommand{\email}[1]{\textsf{#1}}

\newcommand{\cpp}[1][]{C\nolinebreak\hspace{-.05em}\raisebox{.4ex}{\relsize{-3}{\textbf{+}}}\nolinebreak\hspace{-.10em}\raisebox{.4ex}{\relsize{-3}{\textbf{+}}}#1}
\newcommand{\gpp}[1][]{g\nolinebreak\hspace{-.05em}\raisebox{.4ex}{\relsize{-3}{\textbf{+}}}\nolinebreak\hspace{-.10em}\raisebox{.4ex}{\relsize{-3}{\textbf{+}}}#1}
\newcommand{\etal}{et al.~}

\makeatletter
\newcommand{\leqnomode}{\tagsleft@true}
\newcommand{\reqnomode}{\tagsleft@false}
\makeatother

\crefname{lp-fig}{}{}
\crefname{lp}{LP}{LPs}
\crefname{ip}{IP}{IPs}
\crefname{definition}{Definition}{Definitions}
\crefname{claim}{Claim}{Claims}
\crefname{observation}{Observation}{Observations}
\crefname{lemma}{Lemma}{Lemmas}
\crefname{figure}{Figure}{Figures}

\newcommand{\ps}{{\small \textsf{PS}}}

\newcommand{\tsmall}{{\sc Small}}
\newcommand{\tlarge}{{\sc Large}}
\newcommand{\tdata}{{\sc Trace}}
\newcommand{\gfive}{{\sc Rmat}}

\newcommand{\algname}[1]{\textsc{#1}}
\newcommand{\stk}{\algname{Stk}}
\newcommand{\stkdp}{\algname{Stk-dp}}
\newcommand{\git}{\algname{Grdy-It}}
\newcommand{\gpait}{\algname{GPA-It}}
\newcommand{\nc}{\algname{NC}}
\newcommand{\kec}{\algname{k-EC}}

\newcommand{\stkbnh}{\algname{Stkb}}
\newcommand{\stkb}{\algname{Stkb-cc-m}}
\newcommand{\stkbccm}{\algname{Stkb-cc-m}}
\newcommand{\stkbcc}{\algname{Stkb-cc}}
\newcommand{\stkbm}{\algname{Stkb-m}}

%% file: sections/abstract.tex
\begin{abstract}
    We design and implement two single-pass semi-streaming algorithms for the maximum weight $k$-disjoint matching (\textsf{$k$-DM}) problem. Given an integer $k$, the \textsf{$k$-DM} problem is to find $k$ pairwise edge-disjoint matchings such that the sum of the weights of the matchings is maximized. For $k \geq 2$, this problem is NP-hard. Our first algorithm is based on the primal-dual framework of a linear programming relaxation of the problem and is $\frac{1}{3+\varepsilon}$-approximate. We also develop an approximation preserving reduction from \textsf{$k$-DM} to the maximum weight $b$-matching problem. Leveraging this reduction and an existing semi-streaming $b$-matching algorithm, we design a $(\frac{1}{2+\varepsilon})(1 - \frac{1}{k+1})$-approximate semi-streaming algorithm for \textsf{$k$-DM}. For any constant $\varepsilon > 0$, both of these algorithms require $O(nk \log_{1+\varepsilon}^2 n)$ bits of space. To the best of our knowledge, this is the first study of semi-streaming algorithms for the \textsf{$k$-DM} problem.
    
   We compare our two algorithms to state-of-the-art offline algorithms on 95 real-world and synthetic test problems, including thirteen graphs generated from data center network traces. On these instances, our streaming algorithms used significantly less memory (ranging from 6$\times$ to 512$\times$ less) and were faster in runtime than the offline algorithms. Our solutions were often within 5\% of the best weights from the offline algorithms. We highlight that the existing offline algorithms run out of 1 TB of memory for most of the large instances ($>1$ billion edges), whereas our streaming algorithms can solve these problems using only 100 GB memory for $k=8$.
\end{abstract}

%% file: sections/intro.tex
\section{Introduction} \label{sec:intro}

Given an undirected graph $G = (V, E, w)$ with weights $w \colon E \to \R_{> 0}$ and an integer $k \geq 1$, the 
\textsf{$k$-Disjoint Matching (\kdm{})} problem asks for a collection of $k$ pairwise edge-disjoint matchings
that maximize the sum of the weights of matched edges.
The \kdm{} problem is a generalization of the classical \textsf{Maximum Weight Matching} (\mwm{}) problem and is closely related 
to the \textsf{Maximum Weight $b$-Matching} (\mwbm{}) problem. 
However, in contrast to these problems, \kdm{} is \NPH{} and \APXH{} already for $k \geq 2$ \cite{feige2002approximating,hanauer2022fast}. 
Prior work has primarily studied \kdm{} in computational models where space complexity is not a limiting factor in designing algorithms.
In this work, we study \kdm{} in the single-pass \emph{semi-streaming} model \cite{FeigenbaumKMSZ05,Muthukrishnan2005}, which is used to solve massive graph problems with limited memory. 
In particular, we extend existing state-of-the-art semi-streaming matching~\cite{PazS19,GhaffariW19} and $b$-matching~\cite{huang2021semi} algorithms to the \kdm{} problem.
To the best of our knowledge, these are the \emph{first} semi-streaming algorithms for the \kdm{} problem.  

In the offline unweighted setting, \kdm{} in general graphs was originally studied by Feige \etal\cite{feige2002approximating}, who 
motivated the problem by applications in scheduling %
traffic in satellite-based communication networks. 
Cockayne \etal\cite{cockayne1978linear} modeled the problem of finding a maximal assignment of jobs to people such that no person performs the 
same job on two consecutive days using unweighted \kdm{} in bipartite graphs with $k=2$. 

In the weighted setting, \kdm{} was %
recently studied by Hanauer \etal\cite{hanauer2022fast,hanauer2023dynamic} in the offline and 
dynamic computation models. This was motivated by applications in designing 
reconfigurable optical topologies for data center networks \cite{ballani2020sirius,bienkowski2021online,bienkowski2022online,mellette2017rotornet}. 
In contrast to %
static networks,  reconfigurable networks use optical switches to quickly provide direct connectivity between racks, where 
each switch essentially acts as a reconfigurable optical matching. Given a traffic matrix and $k$ optical switches, the underlying optimization problem %
becomes how to compute heavy disjoint matchings that carry a large amount of traffic for each switch, which is exactly the \kdm{} problem. 

\paragraph{Algorithmic Contributions} 
We provide a primal-dual linear programming (LP) formulation of the \kdm{} problem and use it to 
derive a $\frac{1}{3 + \varepsilon}$-approximate single-pass semi-streaming algorithm that requires $O(nk \log^2 n)$ bits of space for any constant $\varepsilon > 0$. %
Our algorithm extends the seminal \textsf{MWM} semi-streaming algorithm by Paz and Schwartzman~\cite{PazS19} by  %
maintaining $k$ stacks and employing approximate dual variables to decide which edges should be stored in those stacks. 
The post-processing phase that computes $k$ edge-disjoint matchings from the stacks is more involved here since edges in a stack 
that are not included in a matching need to be considered for inclusion in higher-numbered stacks. The primal-dual analysis of the 
approximation ratio involves two sets of dual variables here, unlike the former algorithm.

We also reduce the \kdm{} problem to the \textsf{MW$b$M} problem. 
In particular, we show that %
a modified edge coloring algorithm on any $\alpha$-approximate $b$-matching subgraph (with $b(v)=k$ for all $v \in V$) computes an 
$\alpha (1-\frac{1}{k+1})$-approximate solution for \kdm{}. 
Using the $\frac{1}{2+\eps}$-approximate semi-streaming \mwbm{} algorithm of Huang and Sellier \cite{huang2021semi}, 
we obtain a $(\frac{1}{2+\varepsilon}) (1-\frac{1}{k+1})$-approximate \kdm{} that requires $\bigO{nk \log^2 n}$ bits of space for any constant $\eps > 0$. 
This reduction, which was previously known for unweighted \kdm{}~\cite{feige2002approximating}, is not specific to the semi-streaming setting, and thus 
could be used to develop algorithms for \kdm{} in other computational models where  
$b$-matching results are known.

\paragraph{Experimental Validation} 
We implement both algorithms and compare the memory used, running time required, and the weight computed with static 
offline approximation algorithms for this problem on several real-world and synthetic graphs, and several graphs 
generated from data center network traces. %
Our results show that the streaming algorithms reduce the memory needed to compute the matchings often by two orders of magnitude 
and are also faster than offline static algorithms. 
Indeed, the latter algorithms do not terminate on all but one of the larger graphs in our test set. 
The median weights computed by the streaming algorithms are only about $5\%$ lower than the ones obtained by the static algorithms. 
Among the streaming algorithms, the primal-dual algorithm outperforms the
$b$-matching-based algorithm in memory needed and weight, and also time (except for the data center problems). 

%% file: sections/preliminaries.tex
\section{Preliminaries} \label{sec:prelims}

\paragraph{Notation}
Consider a graph $G=(V, E, w)$ with weights $w \colon E \to \R_{> 0}$.
We denote $n \coloneqq  \abs{V}$ and $m \coloneqq \abs{E}$ throughout the paper.
For an edge $e = (u, v)$, we say that vertices $u$ and $v$ are \emph{incident} on the edge $e$. 
Given a vertex $v \in V$, we denote by $\delta(v)$ the set of edges $v$ is incident on,  and 
by $\degree{v} \coloneqq \abs{\delta(v)}$ its degree. The maximum degree of $G$ is  $\Delta \coloneqq \max_{v \in V} \degree{v}$. 
We say that two edges $e_1$ and $e_2$ are \emph{adjacent} if they share a common vertex. 
For an edge subset $H \sse E$, we let $V(H)$  denote the set of vertices incident on edges in $H$, 
and let $G[H]$ denote the subgraph \emph{induced} by $H$ (i.e., the subgraph whose edge set is $H$
and vertex set is $V(H)$). 
Likewise, we denote by  $\degree[H]{v} \coloneqq \abs{\delta(v) \cap H}$  the 
number of edges in $H$ that a vertex $v \in V$ is incident on and let $\Delta_H \coloneqq \max_{v \in V} \degree[H]{v}$. 
For a positive integer $t$, we use $[t]$ to represent the set of integers from $1$ to $t$, inclusive.  
For an integer $s \leq t$, we let $[s .. t]$ denote the set of integers from $s$ to $t$, inclusive.  

\paragraph{Matchings and $b$-Matchings} 
Given a function $b: V \to \Z_+$, a \emph{$b$-matching} in a graph $G$ is an edge subset $F \sse E$ 
such that $\abs{F \cap \delta(v)} \leq b(v)$ for all $v \in V$. The weight of a $b$-matching $F$ 
is $w(F) \coloneqq \sum_{e \in F} w(e)$, and in the \textsf{Maximum Weight $b$-Matching} (\mwbm{}) problem, 
we aim to maximize $w(F)$. When $b(v)=1$ for all $v \in V$, we obtain a matching and the \mwbm{} problem 
reduces to the \textsf{Maximum Weight Matching} (\mwm{}) problem. 

\paragraph{$k$-Disjoint Matchings}
Given an integer $k \geq 1$, a $k$-disjoint matching in $G$ is a collection of $k$ matchings $\CalM = \{M_1,\ldots, M_k\}$ 
that are pairwise edge-disjoint (i.e.,~$M_i \cap M_j = \emptyset$ for all $i, j \in [k], i \neq j$). 
Its weight is given by $w(\CalM) \coloneqq \sum_{i=1}^k w(M_i)$ and in the \textsf{$k$-Disjoint Matching (\kdm{})} problem, we aim %
to maximize $w(\calM)$. 
A $k$-disjoint matching can also be described through an edge coloring
viewpoint \cite{hanauer2022fast}. Consider a coloring function $\calC \colon E \to [k] \cup \lrc{\perp}$ that assigns edges a color from the 
palette $[k]$, or leaves them uncolored (color $\perp$). 
If $\calC$ describes a proper $k$ edge coloring (i.e., any two adjacent edges $e_1, e_2$ colored from $[k]$
satisfy $\calC(e_1) \neq \calC(e_2)$) 
then it also describes a $k$-disjoint matching. Prior work has shown that \kdm{} is \NPH{} and \APXH{} for $k \geq 2$~\cite{feige2002approximating,hanauer2022fast}. 

An LP relaxation of the \kdm{} problem and its dual is shown in \primal{} and \dual{}, respectively. For each edge $e = (u,v) \in E$ and color $c \in [k]$, 
we associate each primal variable $x(c,e)$ with the inclusion of edge $e$ in the $c^{\text{th}}$ matching, i.e., $x(c,e) = 1$ iff $e \in M_c$. 
The first constraint in \primal{} enforces that $M_c$ is a valid matching for each $c \in [k]$, while the second constraint ensures each edge $e \in E$ belongs to at most one matching. For the dual \dual{}, we define variables 
$y(c, v)$ for each color $c \in [k]$ and vertex $v \in V$ (corresponding to the first constraint in \primal{}), and 
$z(e)$ for each edge $e \in E$ (corresponding to the second constraint in \primal{}). 

\input{sections/lp_kDM.tex}

\paragraph{Semi-Streaming Model} 
For semi-streaming \kdm{}, in each pass, the edges of $E$ are presented one at a time in an \emph{arbitrary order}. 
We aim to compute a $k$-disjoint matching in $G$ at the end of the algorithm, using limited memory and only a single pass. 
The semi-streaming model allows memory size for processing proportional (up to polylog factors) to the size of the memory needed to store the output.  
For \kdm{}, the final solution size is $\bigO{nk}$, and hence the memory limit is $\bigO{nk \cdot \polylog{n}} \eqqcolon \bigOlog{nk}$ bits of space. 
We assume that the ratio $W = w_{\max}/w_{\min}$ is $\poly{n}$, where $w_{\max} = \max_{e \in E} \lrc{w(e)}$ and $w_{\min} = \min_{e \in E} \lrc{w(e)}$. 
This allows for storing edge weights and their sums in $\bigO{\log n}$ bits.

%% file: sections/lp_kDM.tex
\begin{figure}[t]
\begin{minipage}[t]{.4\columnwidth}
\small
\begin{maxi*}
    {}{\sum_{c \in [k]} \sum_{e \in E} w(e) x(c, e)}{}{\hypertarget{primal-k-djm}{\mathsf{(P)}}}{} %
    \addConstraint{\sum_{e \in \delta(v)} x(c,e)}{\leq 1}{\; \forall v \in V, c \in [k]}
    \addConstraint{\sum_{c \in [k]}  x(c,e)}{\leq 1}{\; \forall e \in E} 
    \addConstraint{ x(c, e)}{\geq 0}{\; \forall e \in E,  c \in [k]}
\end{maxi*}
\end{minipage}%
\begin{minipage}[t]{.6\columnwidth}
\small
\begin{mini*}
    {}{\sum_{c \in [k]}\sum_{v \in V} y(c, v) + \sum_{e \in E} z(e)}{}{\hypertarget{dual-k-djm}{\mathsf{(D)}}}{} %
    \addConstraint{y(c, u) + y(c,v) + z(e)}{\geq w(e)}{\; \forall e = (u,v) \in E, c \in [k]}
    \addConstraint{y(c, v)}{\geq 0}{\; \forall v \in V, c \in [k]} 
    \addConstraint{z(e)}{\geq 0}{\; \forall e \in E}
\end{mini*}
\end{minipage}

\caption{LP Relaxation \protect\primal{} of \kdm{} and its dual \protect\dual{}.}
\label{fig:lp_formulations_kdm}
\end{figure}

%% file: sections/related_work.tex
\section{Related Work} \label{sec:related}

In this section we discuss relevant related work to semi-streaming \kdm{}, including offline approximation algorithms for \kdm{} and 
results for matching problems in the semi-streaming model. For discussions on practical applications of \kdm{}, including a more in depth explanation of its relevance to reconfigurable datacenter networks, as well as practical offline approximation algorithms for the \mwm{} and \mwbm{} problems, we refer to \cref{sec:kdm-app-offline}. 

\paragraph{Offline Approximation Algorithms}
In the offline setting, Hanauer \etal~\cite{hanauer2022fast} designed six approximation algorithms for \kdm{}. 
Three of these algorithms are based on an iterative matching framework
where $k$ matchings are successively computed by running a matching algorithm and removing the matched edges from the graph.
This framework was used with the Blossom \cite{edmonds1965paths} algorithm, which computes an exact \mwm{} solutions, 
and the Greedy and Global Path \cite{maue2007engineering} algorithms, which compute
$\frac{1}{2}$-approximate \mwm{} solutions. 
They also designed a $ b$-matching-based algorithm, where a Greedy $(k-1)$-matching is first found and 
then converted in a $k$-disjoint matching using the Misra-Gries edge coloring algorithm \cite{misra1992constructive}. 
Additionally, two direct algorithms, NodeCentered and $k$-Edge Coloring, which do not use matching algorithms as a subroutine 
were also proposed. The NodeCentered algorithm assigns ratings to vertices, %
which are then processed in rating-decreasing order, and up to $k$ edges a vertex is incident on are colored with any available color in weight-decreasing order. 
A threshold $\theta \in [0,1]$ is also introduced, which avoids an overly Greedy approach by deferring the coloring of edges 
with weight less than $\theta w_{\max}$.
The $k$-Edge Coloring algorithm is an adaption of the Misra-Gries $(\Delta+1)$ edge coloring algorithm \cite{misra1992constructive} that 
is restricted to using $k$ colors and accounts for edge weights. 
The iterative GPA, $b$-matching based, NodeCentered, and $k$-Edge Coloring algorithms are shown to be \emph{at most} $\frac{1}{2}$-approximate, 
while the Blossom variant is shown to be at most $\frac{7}{9}$-approximate and the Greedy variant is $\frac{1}{2}$-approximate.

\paragraph{Matchings in the Semi-Streaming Model} 

Matching problems are an active area of research in the semi-streaming model. 
For \mwm{} in the single pass, arbitrary order stream setting, Feigenbaum \etal\cite{FeigenbaumKMSZ05} first 
gave a $\frac{1}{6}$-approximation algorithm. This was improved on by a series of papers \cite{CrouchS14, EpsteinLMS11,McGregor05, Zelke12}, 
until the current state-of-the-art result by Paz and Schwartzman \cite{PazS19} who showed that a simple local-ratio algorithm achieves a $\frac{1}{2+\varepsilon}$-approximation. 
Ghaffari and Wajc \cite{GhaffariW19} further simplified the analysis of this algorithm by giving both a primal-dual and charging-based analysis. 
This algorithm was implemented recently by Ferdous \etal\cite{Ferdous+:2024a} and it was shown to reduce memory requirements by one to two orders of magnitude over offline $\frac{1}{2}$-approximate algorithms, while being close to the best of them in run time and matching weight.
On the hardness front, 
Kapralov \cite{kapralov2021space} showed that no single-pass semi-streaming algorithm can have an approximation ratio 
better than $\frac{1}{1 + \ln 2} \approx 0.59$ in arbitrary order streams. 
In random order streams, Gamlath \etal\cite{GamlathKMS19} 
designed a $(\frac{1}{2}+c)$-approximate algorithm, where $c > 0$ is some absolute constant. 

For \mwbm{} in the single pass, arbitrary order stream setting, 
Levin and Wajc \cite{levin2021streaming} designed a $\frac{1}{3 + \varepsilon}$-approximate 
algorithm using a primal-dual framework, which was recently improved to $\frac{1}{2 + \varepsilon}$ by Huang and Sellier \cite{huang2021semi}. 
A variant of the latter algorithm requires $\bigOlog{|F_{\max}| \log_{1+\varepsilon} (\nicefrac{W}{\varepsilon})}$ bits, where $\abs{F_{\max}}$ is the size of 
a max cardinality $b$-matching in $G$.

\paragraph{Edge Colorings and Unweighted \kdm{}}
The \kdm{} problem is equivalent to 
a weighted variant of the \textsf{Edge Coloring} problem; 
in the latter, the goal is to find the \emph{chromatic index} of a graph, i.e., the minimum number of colors needed such that adjacent edges receive distinct colors.
Vizing \cite{vizing1965chromatic} showed that the chromatic index of any simple graph $G$ is in $\lrc{\Delta, \Delta+1}$, but it is \NPH{} to decide between them \cite{holyer1981np}. %
Hence, most edge coloring algorithms, like the 
$\bigO{nm}$ time Misra-Gries algorithm \cite{misra1992constructive}, construct $(\Delta+1)$-edge colorings.
The \kdm{} problem can be seen as a ``maximization'' variant of \textsf{Edge Coloring}, where given the number of colors $k$ as input, 
the goal is to find a maximum weight subgraph with chromatic index $k$. 

Using this coloring viewpoint,
Feige \etal\cite{feige2002approximating} provided several hardness results and approximation algorithms for unweighted \kdm{} in the offline setting,
which was later improved by Kami{\'n}ski and Kowalik \cite{kaminski2014beyond} for small $k$. 
Favrholdt and Nielson \cite{favrholdt2003online} additionally gave algorithms for this problem in the online setting. 
Recently, El-Hayek \etal\cite{elhayek2023dynamic} 
developed fully dynamic unweighted \kdm{} algorithms by reducing it to dynamic $b$-matching followed by edge coloring.

%% file: sections/primal_dual.tex
\section{A Primal-Dual  Approach} \label{sec:stk}

In this section we extend the streaming algorithm of Paz and Schwartzman (henceforth, \ps{}) \cite{PazS19}, and more specifically the primal-dual interpretation of it by Ghaffari and Wajc \cite{GhaffariW19}, 
for the \textsf{MWM} problem to the  \kdm{} problem.
We begin with an intuitive description of the \ps{} algorithm; for a more formal description, see \cref{subsec:alg_streamMatch}. 

Consider the non-streaming setting first. The algorithm chooses an edge with positive weight, includes it in a stack for candidate 
matching edges, and subtracts its weight from neighboring edges. It repeats this process as long as edges with positive weights remain. 
At the end, we unwind the stack and greedily add edges in the stack to the matching. This means that once an edge is added to the matching, any 
neighboring edges in the stack cannot be added to the matching. 

To adapt the algorithm to the streaming setting, an approximate dual variable $\phi(v)$ is kept for each vertex $v$ that accumulates the weights of the edges incident on $v$ that are added to the stack. 
When an edge arrives, we subtract the sum of the $\phi(\cdot)$ variables of the endpoints of the edge from its weight. If this reduced weight is positive, it is added to the stack; otherwise, it is discarded. The rest of the algorithm proceeds as in the non-streaming setting. 
To bound the size of the stack to $O(n \log{n})$, we need one more idea, which is to add an edge $e = (u, v)$ to the stack only if its weight is greater than $(1+ \varepsilon)(\phi(u) + \phi(v))$, for a small constant $\varepsilon > 0$. This 
ensures that neighboring edges added to the stack have weights that increase exponentially in $(1 + \varepsilon)$. It can be shown that if the edge weights are polynomial in $n$, then the size of the stack is bounded as desired and that the approximation ratio becomes $\frac{1}{2 + \varepsilon}$.

We adapt this general idea to develop our algorithm for \kdm{} in \cref{alg:stk-djm}. For each color $c \in [k]$, we maintain a stack $\calS(c)$ that stores the eligible edges for the $c^{\text{th}}$ matching. 
A matching $M_c$ is then greedily computed from each stack $\calS(c)$ in the post-processing phase. 
The algorithm maintains \emph{approximate} dual variables $\phi(c, v)$ for each color $c \in [k]$ and $v \in V$, and uses $\eps > 0$ to process only sufficiently heavy edges.
For an edge $e= (u,v)$ in the stream, we iterate over the colors $c \in [k]$ to verify whether $w(e) \geq (1+\eps)\lrp {\phi(c, u) + \phi(c,v)}$. 
If the condition is not satisfied for any color, then the edge is discarded.  
Otherwise, let $\ell$ be the first color that satisfies it.  The algorithm computes a reduced weight $w'(\ell,e)$ for $e$ by subtracting the sum $\phi(\ell, u) + \phi(\ell, v)$ from its weight $w(e)$, 
pushes $e$ into $\calS(\ell)$, and increases $\phi(\ell, u)$ and $\phi(\ell, v)$ by the reduced weight $w^\prime(\ell, e)$.

In the post-processing phase, each stack $\calS(c)$ is processed in increasing order of the color $c$, and the edges in each stack are processed in reverse order in which they were added (i.e., by popping from the stack). 
For an edge $e=(u,v)$ popped from $\calS(c)$, if no earlier popped edge from $\calS(c)$ is incident on either $u$ or $v$ in $M_c$, 
then $e$ is added to $M_c$. 
Otherwise, the algorithm checks to see if $e$ can be added to a later stack $\calS(j)$ where $j > c$,
again based on the condition that $w(e) \geq (1+\eps)\lrp {\phi(j, u) + \phi(j,v)}$.
At termination, the algorithm returns a $k$-disjoint matching $\calM = \{M_1, \ldots, M_k\}$. 

\input{sections/kDM_pseudocode.tex}

\subsection{Analysis of the Algorithm}
We prove the approximation ratio of \cref{alg:stk-djm} 
using the standard primal-dual framework and adapting the analysis in \cite{GhaffariW19}. We first show how to derive a 
feasible dual solution for the dual \dual{} from the $\phi(\cdot, \cdot)$ values. By weak duality, 
the resulting dual objective immediately provides an upper bound on the weight of an optimal \kdm{} solution. %
\cref{lem:dual_vertex_bound,lem:dual_edge_bound}
then show lower bounds between the value of the $k$-disjoint matching $\calM$ constructed 
by \cref{alg:stk-djm} and the dual variables, 
which are then used to prove that $\calM$ is $\frac{1}{3 + 2\varepsilon}$-approximate
in \cref{thm:stk_approx}. We also prove the space complexity of the algorithm in \cref{lem:stk_space}. 

\subsubsection{Dual Feasibility}
At termination, we set $y(c,v) = (1+\eps) \: \phi(c,v)$ for all $c \in [k]$ and $v \in V$. 
Recall that $y(c,v)$ is a dual variable from \dual{}, 
and $\phi(c,v)$ is an approximate dual variable used in \cref{alg:stk-djm}.
Unlike in classical \mwm{}, %
for \kdm{}, we have to satisfy the dual constraints of each edge for all $c \in [k]$. 
Although the dual variables $z(\cdot)$ are unused in the algorithm, they help ensure dual feasibility; see below.
If an edge $e=(u,v)$ is not in any matching (i.e, $e$ is discarded either in the streaming 
or post-processing phase) then $y(c,u) + y(c,v) \geq w(e)$ for all $c \in [k]$, which satisfies the constraint. 
However, if $e \in M_\ell$ for some $\ell \in [k]$, the dual constraints for $c \in [\ell+1..k]$ may be violated. 
Thus, we set 
\begin{equation} 
z(e) = \max \lrc{0,\; \max_{c \in [k]} \lrc{w(e) - (1+\eps) \lrp{\phi(c,u) + \phi(c,v)}} }.  
\label{eq:z} 
\end{equation}
The following claim is immediate. 

\begin{claim}
    \label{claim:k-djm-feas}
    For all vertices $v \in V$, edges $e \in E$, and $c \in [k]$, the dual variables $y(c,v)$ and $z(e)$
    defined above constitute a feasible solution to \dual{}. 
\end{claim}

\subsubsection{Approximation Ratio} 
To prove the approximation ratio, we first separately relate the weight of the solution returned 
by \cref{alg:stk-djm} to the summations of the 
$\phi(\cdot, \cdot)$ and $z(\cdot)$ variables. 

\begin{lemma}
    \label{lem:dual_vertex_bound}
    The solution $\calM$ output by \cref{alg:stk-djm} satisfies 
    $w(\calM) \geq \frac{1}{2} \sum_{c \in [k]} \sum_{v \in V} \phi(c, v)$. 
\end{lemma}
\begin{proof}
    It suffices to show that $w(M_c) \geq \frac{1}{2} \sum_{v \in V} \phi(c, v)$, for any matching $M_c \in \calM$.
    Let $E_c$ be the set of edges that were pushed to the stack $\calS(c)$ at some point in either the streaming (line \ref{algline:cond1}) or the 
    post-processing (line \ref{algline:cond2}) phases. 
    Note that only edges in $E_c$ could have caused the $\phi(c, \cdot)$ values to increase. 
    For ease of analysis, for an edge $e^\prime = (s, t) \in E_c$ let $\phi^{\text{old}}_{e^\prime}(c, \cdot)$ and $\phi^{\text{new}}_{e^\prime}(c, \cdot)$ 
    denote the $\phi(c, \cdot)$ values before and after $e^\prime$ is pushed to $\calS(c)$, respectively. 
    By definition of how we update the $\phi(c, \cdot)$ values, we have 
    $\phi_{e^\prime}^{\text{new}}(c, s) = \phi_{e^\prime}^{\text{old}}(c, s) + w^\prime(c, e^\prime)$, $\phi_{e^\prime}^{\text{new}}(c, t) = \phi_{e^\prime}^{\text{old}}(c, t) + w^\prime(c, e^\prime)$, 
    and $\phi_{e^\prime}^{\text{new}}(c, r) = \phi_{e^\prime}^{\text{old}}(c, r)$ for all $r \in V \sm \lrc{s, t}$.  
    This implies 
    \begin{equation}
        w^\prime(c, e^\prime) = \frac{1}{2} \sum_{x \in e^\prime} \phi_{e^\prime}^{\text{new}}(c, x) - \phi_{e^\prime}^{\text{old}}(c, x).
        \label{eq:diff_reduced_weight}
    \end{equation}
    Upon termination of \cref{alg:stk-djm}, since initially 
    $\phi(c,v) = 0$ for all $v \in V$, we also have that
    \begin{equation}
        \phi(c, v) = \sum_{e^\prime \in E_c} \phi_{e^\prime}^{\text{new}}(c, v) - \phi_{e^\prime}^{\text{old}}(c, v). 
        \label{eq:diff_phi}
    \end{equation}
    Now for an edge $e = (u, v) \in M_c$,  let 
    \begin{equation*}
        \calP_{<}(c, e) \coloneqq  \lrc{e^\prime \in E_c \colon e \cap e^\prime \neq \emptyset, e^\prime \text{ added to } \calS(c) \text{ before } e},
    \end{equation*} 
    i.e., the set of edges adjacent to $e$ that were pushed to $\calS(c)$ before $e$ was, and let $\calP(c, e) \coloneqq \calP_{<}(c, e) \cup \lrc{e}$. 
    Note that since we construct $M_c$ greedily, no edge $e^\prime \in \calP_{<}(c, e)$ is included in $M_c$ and 
    $E_c = \bigcup_{e \in M_c} \calP(c, e)$. 
    By definition of how we update the $\phi(c, \cdot)$ values, we have that
    $\phi_e^{\text{old}}(c, u) + \phi_e^{\text{old}}(c, v) = \sum_{e^\prime \in \calP_{<}(c, e)} w^\prime(c, e^\prime)$.  
    Additionally, by the definition of $w^\prime(c, e)$,
    \begin{align*}
        w(e) &= w^\prime(c, e) + \phi_e^{\text{old}}(c, u) + \phi_e^{\text{old}}(c, v)
        = \sum_{e^\prime \in \calP(c, e)} w^\prime(c, e^\prime) 
        = \frac{1}{2}\sum_{e^\prime \in \calP(c, e)} \sum_{x \in e^\prime} \phi_{e^\prime}^{\text{new}}(c, x) - \phi_{e^\prime}^{\text{old}}(c, x),
    \end{align*}
    where the last equality follows by \cref{eq:diff_reduced_weight}. 
    Hence, 
    \begin{align*}
        w(M_c) = \sum_{e \in M_c} w(e) 
        &= \frac{1}{2} \sum_{e \in M_c}  \sum_{e^\prime \in \calP(c, e)} \sum_{x \in e^\prime} \phi_{e^\prime}^{\text{new}}(c, x) - \phi_{e^\prime}^{\text{old}}(c, x) \\ 
        &\geq \frac{1}{2} \sum_{e \in E_c}  \sum_{v \in e} \phi_{e}^{\text{new}}(c, v) - \phi_{e}^{\text{old}}(c, v) \\ 
        &= \frac{1}{2} \sum_{v \in V} \sum_{e \in E_c} \phi_{e}^{\text{new}}(c, v) - \phi_{e}^{\text{old}}(c, v)
        = \frac{1}{2} \sum_{v \in V} \phi(c, v). 
    \end{align*}
    The inequality follows since each edge $e^\prime = (u, v) \in E_c \sm M_c$ appears in at least one and at most two $\calP(c, \cdot)$ sets (say, if there exists
    $e_1, e_2 \in M_c$ that $u$ and $v$ are incident on, respectively)
    and the last equality follows by \cref{eq:diff_phi}.
\end{proof}

\begin{lemma}
    The solution $\calM$ output by \cref{alg:stk-djm} satisfies 
    $w(\calM) \geq \sum_{e \in E} z(e)$. 
    \label{lem:dual_edge_bound}
\end{lemma}
\begin{proof}
    From the definition of $z(\cdot)$ in \cref{eq:z}, 
    we have $z(e) \leq w(e)$ for all $e \in E$. Moreover, we can show that $z(e) = 0$ for each edge $e = (u, v) \notin \calM$.
    This holds since either $e$ was discarded during the streaming phase, or during the post-processing phase. 
    In either case,  $w(e) < (1+\varepsilon)(\phi(c, u) + \phi(c, v))$ for all $c \in [k]$, which gives $z(e) = 0$. 
    Hence, $\sum_{e \in E} z(e) =  \sum_{e \in \calM} z(e) + \sum_{e \in E \sm \calM} z(e) \leq \sum_{e \in \calM} w(e) = w(\calM)$. 
\end{proof}

Using \cref{lem:dual_vertex_bound,lem:dual_edge_bound} and weak duality, we can now show the approximation ratio. 

\begin{theorem}
    For any constant $\varepsilon > 0$, the $k$-disjoint matching $\calM$ returned by \cref{alg:stk-djm} is a 
    $\frac{1}{3 + 2\varepsilon}$-approximate solution to \kdm{}.  
    \label{thm:stk_approx}
\end{theorem}
\begin{proof}
    Let $\calM^{\ast}$ be an optimal solution to \kdm{}. By weak duality and the 
    fact that \primal{} is an LP-relaxation of \kdm{}, we have that 
    $w(\calM^{\ast}) \leq \sum_{c \in [k]} \sum_{v \in V} y(c, v)  + \sum_{e \in E} z(e)$
    for the dual variables $y(\cdot, \cdot)$ and $z(\cdot)$ defined in \cref{claim:k-djm-feas}. 
    Recalling that we set $y(c, v) = (1 + \varepsilon)\phi(c, v)$, \cref{lem:dual_vertex_bound,lem:dual_edge_bound}
    imply that
    $(2(1 + \varepsilon)) w(\calM) \geq \sum_{c \in [k]} \sum_{v \in V} y(c, v)$ and $w(\calM) \geq \sum_{e \in E} z(e)$, respectively. 
    Combining these, we obtain 
    \begin{align*}
        (3 + 2\varepsilon) w(\calM) &\geq \sum_{c \in [k]} \sum_{v \in V} y(c, v) + \sum_{e \in E} z(e) \geq w(\calM^{\ast}),
    \end{align*}
    which when rearranged gives $w(\calM) \geq \frac{1}{3 + 2\varepsilon} w(\calM^{\ast})$. 
\end{proof}

\subsubsection{Time and Space Complexity}
The total runtime of \cref{alg:stk-djm} is $O(km)$, which follows as the processing time for each edge is $O(k)$ as it may be considered for insertion %
into each of the $k$ stacks. Additionally, the size of each stack is trivially bounded by $m$, so the post-processing step of unwinding the stacks takes $\bigO{km}$ time. 
The space complexity of \cref{alg:stk-djm} can also easily be bound. 
We first make the following useful observation. 

\begin{observation}
    When an edge $e = (u, v)$ gets pushed to a stack $\calS(c)$, both 
    $\phi(c, v)$ and $\phi(c, u)$ increase by at least a factor of $1 + \varepsilon$. 
    \label{obs:phi_increase_factor}
\end{observation}
\begin{proof}
    Let $\phi_{e}^{\text{old}}(c, \cdot)$ and $\phi_{e}^{\text{new}}(c, \cdot)$ 
    be the values of $\phi(c, \cdot)$ before and after $e$ is pushed to $\calS(c)$, respectively.
    Note that since $e$ is pushed to $\calS(c)$, it must be that $w(e) \geq (1+\varepsilon)\Phi^{\text{old}}_e$, where 
    $\Phi^{\text{old}}_e \coloneqq \phi^{\text{old}}_{e}(c, u) + \phi_{e}^{\text{old}}(c, v)$.  
    Additionally, by how we update the $\phi(c, \cdot)$ values, we have $\phi_e^{\text{new}}(c, v) - \phi_e^{\text{old}}(c, v) = w'(c, e) = w(e) - \Phi^{\text{old}}_e$. 
    Thus, 
    \begin{align*}
            \phi^{\text{new}}_{e}(c, v) - \phi^{\text{old}}_{e}(c, v) = w(e) - \Phi^{\text{old}}_e 
            \geq (1+\varepsilon)\Phi^{\text{old}}_e - \Phi^{\text{old}}_e 
            \geq \varepsilon \phi^{\text{old}}_e(c, v). 
    \end{align*} 
    Rearranging, we get $\phi^{\text{new}}_{e}(c, v) \geq (1+\varepsilon) \phi^{\text{old}}_{e}(c, v)$. 
    The same argument holds for vertex $u$. 
\end{proof}

\begin{lemma}
    For any constant $\varepsilon > 0$, \cref{alg:stk-djm} uses 
    $\bigO{nk \log^2 n}$ bits of space.
    \label{lem:stk_space}
\end{lemma}     
\begin{proof}
    Consider a vertex $v \in V$ and color $c \in [k]$. 
    Let $e = (u, v)$ be an edge that is pushed to $\calS(c)$, and let 
     $\phi_{e}^{\text{old}}(c, \cdot)$ and $\phi_{e}^{\text{new}}(c, \cdot)$ denote
    the values of $\phi(c, \cdot)$ before and after $e$ is pushed to $\calS(c)$, respectively.
    Suppose that after $e$ is pushed, we have that $v$ is incident on $d$ edges in $\calS(c)$. 
    For the special case of $d = 1$
    corresponding to the first edge $v$ is incident on that is included in $\calS(c)$, we can derive a lower bound 
    on $\phi_e^{\text{new}}(c,v)$.
    We use  $\phi_e^{\text{old}}(c, v) = 0$ and $w(e) \geq (1 + \varepsilon)\phi_{e}^{\text{old}}(c, u)$ to obtain
    \begin{equation*}
        \phi_e^{\text{new}}(c, v) = w'(c, e) = w(e) - \phi_e^{\text{old}}(c, u) \geq w(e) - \frac{w(e)}{1 + \varepsilon} \geq \frac{\varepsilon w_{\min}}{1 + \varepsilon}. 
    \end{equation*}
    That is, the minimum non-zero value of $\phi(c, v)$ is at least $\frac{\varepsilon w_{\min}}{1 + \varepsilon}$. Using this together with \cref{obs:phi_increase_factor} implies that for arbitrary values of $d$, 
    $\phi_e^{\text{new}}(c, v) \geq \frac{\varepsilon w_{\min}}{1 + \varepsilon} (1 + \varepsilon)^{d-1}$. Moreover, by definition of how we compute reduced weights and update the $\phi(c, \cdot)$ values, we have that $\phi_e^{\text{new}}(c, v) \leq w_{\max}$. Recalling that $W = \frac{w_{\max}}{w_{\min}}$ and using these two bounds, we find that 
    $(1 + \varepsilon)^{d-2} \leq {W}{\varepsilon}^{-1}$. 
    Taking the logarithm of both sides, we get 
    \begin{equation*}
        d \leq 2 + \log_{1 + \varepsilon} (W\varepsilon^{-1}) = \bigO{\log n},
    \end{equation*}
    since we assume $\varepsilon$ is constant and $W$ is $\poly{n}$. 
    That is, $v$ can be incident on at most $\bigO{\log n}$ edges in $\calS(c)$. 
    Hence, $\abs{\calS(c)} = \bigO{n \log n}$ %
    and the total number of edges stored in all the stacks is $\bigO{nk \log n}$. 
    Each edge weight requires $\bigO{\log n}$ bits; similarly, each $\phi(\cdot, \cdot)$ variable requires $\bigO{\log n}$ bits as it is the sum of 
    at most $\Delta < n$ edge weights, giving the space complexity of $\bigO{nk\log^2 n}$ bits. 
\end{proof}

%% file: sections/kDM_pseudocode.tex
\begin{algorithm}[!t]

\caption{Semi-Streaming \kdm{}} 
\label{alg:stk-djm}

\quad\; \textbf{Input:} A stream of edges $E$, an integer $k$, and a constant $\eps > 0$ \\
\indent \quad\; \textbf{Output:} A \mapprox-approximate $k$-disjoint matching $\calM$ using $\bigO{nk\log^2{n}}$ bits of space\\
\begin{minipage}[t]{0.5\linewidth}
\begin{algorithmic}[1]
\LComment{Initialization}
\State{$\forall v\in V, \forall c \in [k]: \phi(c,v) \gets 0$}

\State{$\CalS \gets \lrc{\CalS_1, \ldots, \CalS_k} $, where $\CalS(c)$ denotes stack $\CalS_c$}
\Statex{}
\LComment{Streaming Phase}
\For{$e = (u,v) \in E$}
    \For{$c \in [k]$}
        \State $\phi_c = \phi(c,u) + \phi(c,v)$
        \If{$w(e) \geq (1+\eps) \phi_c$} \label{algline:cond1}        
            \State{$w^\prime(c, e) \gets w(e) - \phi_c$} \label{algline:red_weight}
            \State{$\phi(c,u) \gets \phi(c,u) + w^\prime(c, e)$}\label{algline:phi_u_update} %
            \State{$\phi(c,v) \gets \phi(c,v) + w^\prime(c, e)$} \label{algline:phi_v_update}
            \State{{$\CalS(c)$.push($e$)}; \textbf{break}}
        \EndIf
    \EndFor
\EndFor
\end{algorithmic}
\end{minipage}
\hfill
\begin{minipage}[t]{0.49\linewidth}
\begin{algorithmic}[1]
\setcounter{ALG@line}{12}
\LComment{Post-Processing}
\State{$\forall c \in [k]:M_c \gets \emptyset$}
\For{$c \in [k]$}
    \While{$\CalS(c)$ is not empty}
        \State{$e = (u,v) \gets \CalS(c)$.pop()}
        \If{$V(M_c) \cap \{u,v\} = \emptyset$} 
            \State{$M_c \gets M_c \cup \{e\}$}
        \Else
            \For{$j \in [c+1..k]$}
                \State $\phi_j = \phi(j,u) + \phi(j,v)$
                \If{$w(e) \geq (1+\eps) \phi_j$} \label{algline:cond2}                
                    \State{$w^\prime(j, e) \gets w(e) - \phi_j$}
                    \label{algline:red_weight1}
                    \State{$\phi(j,u) \gets \phi(j,u) + w^\prime(j, e)$} \label{algline:phi_u_update1}
                    \State{$\phi(j,v) \gets \phi(j,v) +w^\prime(j, e)$} \label{algline:phi_v_update1}%
                    \State{{$\CalS(j)$.push($e$)}; \textbf{break}}
                \EndIf
            \EndFor
        \EndIf
    \EndWhile
\EndFor
\State{\Return $\calM = \lrc{M_1, \ldots, M_k}$}
\end{algorithmic}
\end{minipage}
\end{algorithm}

%% file: sections/bmatching.tex
\section{A \texorpdfstring{$b$}{b}-Matching Based Approach}
\label{sec:bmatching}

Recall that a $b$-matching generalizes a matching by allowing each vertex to be 
incident to \emph{at most} $b(v)$ matched edges for some function $b \colon V \to \Z_+$. 
When $b(v) = k$ for all $v \in V$, where $k$ is some positive integer, 
we refer to the matching as a $k$-matching and consider the \textsf{Maximum Weight $k$-Matching} (\mwkm{}) problem.  
Note that 
$k$-disjoint matchings always induce valid $k$-matchings, but the reverse need not hold (e.g., the triangle graph with $k=2$).  
In this sense, \mwkm{} provides a relaxation of \kdm{}
(i.e., if $F^{\ast}$ and $\calM^{\ast}$ are optimal solutions to \mwkm{} and \kdm{} on the same graph, respectively, then $w(F^{\ast}) \geq w(\calM^{\ast})$). 
This leads to the following approach to construct a feasible $k$-disjoint matching:

\begin{enumerate}[]
    \item Solve \mwkm{} on the graph $G$, which gives a $k$-matching $F$.  
    Note that $\Delta_F$, the maximum degree of a vertex in the induced graph  $G[F]$, may be less than $k$.
    \item Properly $(\Delta_F + 1)$-edge color the subgraph $G[F]$, which may use up to $k+1$ colors.
    \item Return $\calM$, the collection of edges colored by the $k$ heaviest color classes.
\end{enumerate}

This approach was originally used for unweighted \kdm{} by Feige \etal\cite{feige2002approximating}, where they showed it provided a $(1- \frac{1}{k+1})$-approximation guarantee. 
Here we extend this to weighted \kdm{} and show that the reduction is approximation preserving.

\begin{restatable}{lemma}{bmatchinglemma}
    Let $F$ be an $\alpha$-approximate solution to \mwkm{} on a graph $G$. If the
    induced subgraph $G[F]$ is properly $(\Delta_F + 1)$ colored, the set of edges colored by the 
    $k$ heaviest color classes is an $ \alpha (1- \frac{1}{k+1})$-approximate solution to \kdm{} on $G$. 
    \label{lem:bmatching_to_kdm}
\end{restatable}

\begin{proof}
    Let $\calM$ represent a solution to \kdm{} on $G$. 
    Additionally, let $F^{\ast}$ and $\calM^{\ast}$ be the optimal solutions to \mwkm{} and \kdm{} on $G$, respectively. 
    
    By definition of a $k$-matching, we have that $\Delta_F \leq k$. 
    If $\Delta_F < k$, then the edge coloring used at most $k$ colors, and we can return $\calM = \lrc{ M_1, \ldots, M_k}$, where $M_i$ is the set 
    of edges colored with $i$ for $i \in [k]$. In this case, we have $w(\calM) = w(F)$. Otherwise, if $\Delta_F = k$, then the edge coloring may have 
    used $k + 1$ colors. Without loss of generality, let $k+1$ denote the color class with the minimum weight. 
    Again let $\calM = \lrc{M_1, \ldots, M_k}$.
    By discarding the edges with color $k+1$, at most a $\frac{1}{k+1}$ fraction of the weight of $F$ is lost. 
    Thus, in either case
    \begin{align*}
        w(\calM) \geq \lrp{1 - \frac{1}{k+1}} w(F)
             \geq \alpha \lrp{1 - \frac{1}{k+1}}  w(F^{\ast}) 
             \geq \alpha  \lrp{1 - \frac{1}{k+1}}  w(\calM^{\ast}),    
    \end{align*}
    where the penultimate inequality follows from the definition of $F$, and the last inequality 
    follows from \mwkm{} being a relaxation of \kdm{}.
\end{proof}

Note that properly $(\Delta+1)$-edge coloring a graph $G$ can be done in $\bigO{m}$ space using the $\bigO{nm}$ time Misra-Gries algorithm \cite{misra1992constructive}. 
If we use a semi-streaming algorithm for \mwkm{} to handle the streaming process and find some $k$-matching $F$, the remaining
steps of the algorithm only require memory linear in $\abs{F}$, resulting in a semi-streaming algorithm for \kdm{}.
Using the semi-streaming $\frac{1}{2+\varepsilon}$-approximation algorithm of Huang and Sellier \cite{huang2021semi} for 
\mwbm{} with $b(v) = k$ for all $v \in V$, 
\cref{lem:bmatching_to_kdm} implies a semi-streaming 
$(\frac{1}{2+\varepsilon}) (1 - \frac{1}{k+1})$-approximation algorithm for \kdm{}. %
The space requirement is $\bigO{nk \log^2 n}$ bits,
and it is determined by the Huang and Sellier algorithm. 
We describe the algorithm formally
in \cref{alg:kmatching-to-kdm}, where  \Call{SS-$b$M}{} and \Call{Color}{} refer to the
algorithms of Huang and Sellier \cite{huang2021semi} and Misra and Gries \cite{misra1992constructive}, respectively. 
For completeness, in \cref{sec:alg_streambmatch,sec:alg_misragries} we give a detailed summary of how these algorithms work. 

\begin{theorem}
    For any constant $\varepsilon > 0$, \cref{alg:kmatching-to-kdm} is a $(\frac{1}{2+\varepsilon}) (1- \frac{1}{k+1})$-approximate 
    semi-streaming algorithm for \kdm{} that uses $\bigO{nk \log^2 n}$ bits of space. 
    \label{thm:stkb}
\end{theorem}

The streaming phase requires $\bigO{k}$ processing time per
edge,  while constructing the $k$-matching $F$ takes $\bigO{m}$ time. By definition of a $k$-matching, $\abs{F} = \bigO{kn}$, so the post-processing 
coloring step requires $\bigO{kn^2}$ time. 
Thus the time complexity of \cref{alg:kmatching-to-kdm} is $\bigO{km + kn^2}$. 

\input{sections/bmatching_pseudocode.tex}

%% file: sections/bmatching_pseudocode.tex
\begin{algorithm}[t]

\caption{Semi-Streaming \kdm{} via \mwkm{}} 
\label{alg:kmatching-to-kdm}

\quad\; \textbf{Input:} A stream of edges $E$, an integer $k$, and a constant $\eps > 0$ \\
\indent \quad\; \textbf{Output:} A $(\frac{1}{2+\varepsilon})(1 - \frac{1}{k+1})$-approximate $k$-disjoint matching $\calM$ using $\bigOlog{nk}$ bits of space\\
\begin{minipage}[t]{0.5\linewidth}
\begin{algorithmic}[1]
    \LComment{Initialization}
    \State{$\forall v \in V \colon b(v) \gets k$}
    \Statex
    \LComment{Streaming Phase}
    \State{$F \gets $ \Call{SS-$b$M}{$E$, $b$, $\frac{\varepsilon}{2}$}} \Comment{\mwkm{}}
\end{algorithmic}
\end{minipage}
\hfill
\begin{minipage}[t]{0.49\linewidth}
\begin{algorithmic}[1]
\setcounter{ALG@line}{4}
\LComment{Post-Processing}
        \State{$\Delta_F \gets \max_{v \in V} \degree[F]{v}$} \Comment{$\Delta_F \leq k$}
        \State{$\calC \gets$ \Call{Color}{$G[F]$}} \Comment{Uses colors $[\Delta_F + 1]$}
        \If{$\Delta_F + 1 = k + 1$} %
            \State{Let $k+1$ be the color class with min weight}
        \EndIf
        \State{$\forall i \in [k] \colon M_i \gets \lrc{e \in F \colon \calC(e) = i}$}
        \State{\Return $\calM = \lrc{M_1, \ldots, M_k}$}
\end{algorithmic}
\end{minipage}
\end{algorithm}

%% file: sections/heuristics.tex
\section{Heuristic Improvements} \label{sec:heur}

In this section, we describe some heuristics we employ
to speed up and improve the weight of both streaming algorithms we have presented.

\paragraph{Dynamic Programming (DP) Based Weight Improvement} \label{subsec:heur_stk}
Manne and Halappanavar \cite{ManneH14} have proposed a general scheme to enhance the weight of a matching by  
computing \emph{two} edge-disjoint matchings $M_1$ and $M_2$. The induced subgraph $G[M_1 \cup M_2]$
contains only cycles of even length or paths. Utilizing a linear-time dynamic programming approach, 
an optimal matching $M^\prime$ can be derived from the induced graph $G[M_1 \cup M_2]$.
The weight of $M^\prime$ is  guaranteed only to be at least as large as $\max \lrc{w(M_1), w(M_2)}$, 
but in practice this heuristic results in substantially improved weight.

We adapted this method for \cref{alg:stk-djm}
as follows: instead of computing a $k$-disjoint matching, 
we first compute a $2k$-disjoint matching. These $2k$ matchings are then merged into $k$ matchings. 
While various strategies can be used for this merging process, we have merged the $i^{\text{th}}$ matching with the $(2k-i+1)^{\text{th}}$ matching, for $i \in [k]$. This approach does not change the asymptotic memory or time complexities for streaming algorithms since each merge requires only \bigO{n} time and space. 

\paragraph{Common Color and Merge} \label{subsec:heur_stkb}
For the $b$-matching based \cref{alg:kmatching-to-kdm}, we used two heuristics. 
The first is the \emph{common color} heuristic 
described by Hanauer \etal\cite{hanauer2022fast}, %
which attempts to color an edge by first determining if there is a common free 
color on both of its endpoints before going through the Misra-Gries routine. 
The second is the \emph{merge} heuristic, 
which is used when the number of color classes is $k+1$; it tries to improve the solution weight by merging the lowest- and second-lowest-weight color classes instead of completely
discarding the lowest-weight one, again %
through the dynamic programming approach described above. 

%% file: sections/experiment.tex
\section{Experiments and Results} \label{sec:experiments}
This section reports experimental results for $95$  real-world and synthetic graphs. 
All the codes were executed on 
a node of a community cluster computer with 128 cores in the node, where the node is an AMD EPYC 7662 with 1 TB of total memory over all the cores. 
The machine has three levels of cache memory. 
The L1 data and instruction caches, the L2 cache, and the L3 cache have 4 MB, 32 MB, and 256 MB of memory, respectively. The page size of the node is 4 KB. 

Our implementation\footnote{Source code: \url{https://github.com/smferdous1/GraST}} uses \cpp{17} and is compiled with \gpp{9.3.0} with the \texttt{-O3} optimization flag. 
The streaming algorithms are simulated by sequentially reading and processing edges from a file using the \cpp{} \texttt{fstream} class. 
We compare them against several offline algorithms in the DJ-Match software suite developed by Hanauer \etal\cite{kathrin_hanauer_2022_dj_match}. 
All the streaming and offline algorithms are sequential, and the reported runtimes \emph{do not} include file reading times 
and (for the offline algorithms) graph construction times. For memory, we use the \texttt{getrusage} system call to report 
the maximum resident set size (RSS) during the program's execution. 

\subsection{Datasets and Benchmark Algorithms}

\paragraph{Real-World and Synthetic Graphs}
Following \cite{hanauer2022fast,khan2018adaptive}, we include ten weighted graphs from the SuiteSparse Matrix Collection \cite{Davis2011} labeled as \tsmall{}. 
Similar to \cite{hanauer2022fast}, we also generated 66 synthetic instances, labeled as \gfive{}, using the R-MAT model \cite{chakrabarti2004rmat} with $2^x$ vertices, 
where $x \in [10, 11, \ldots, 20]$. We used three initiator matrices, $\mathsf{rmat_{\text{b}}} = (0.55, 0.15, 0.15, 0.15)$, 
$\mathsf{rmat_{\text{g}}} = (0.45, 0.15, 0.15, 0.25)$, and $\mathsf{rmat_{\text{er}}} = (0.25, 0.25, 0.25, 0.25)$. 
For all these graphs, we assign real-valued random weights in the range $[1, 2^{19}]$ drawn from uniform or exponential distributions. 
Our \tlarge{} dataset consists of six of the \emph{largest} undirected graphs in the SuiteSparse Matrix Collection~\cite{Davis2011}, each having more than 1 billion edges. For %
the unweighted graphs, we assign uniform random real weights in the range $[1,10^{6}]$. 
In \cref{subsec:dataset}, we list the sizes and degree measures of these graphs in 
\cref{tab:dataset,tab:rmat_graph_stats}. 

\paragraph{Network Trace Data} 
Similar to~\cite{hanauer2022fast}, our network trace (\tdata{}) dataset consists of 
\begin{enumerate*}[noitemsep, label=\roman*)]
    \item Facebook Data Traces~\cite{roy2015inside}: Six production-level traces of three clusters from Facebook’s Altoona Data Center, 
    \item HPC Data~\cite{avin2020complexity}: MPI traces for four different applications run in parallel, 
    \item pFabric Data~\cite{alizadeh2013pfabric,avin2020complexity}: Three synthetic pFabric traces generated from Poisson processes with flow rates in \{0.1, 0.5, 0.8\}.
\end{enumerate*}
From these trace data, we pre-compute graphs by assigning the total demand of a pair of nodes (i.e., the number of times they appear in the trace) as the edge weight. 
In \cref{subsec:dataset}, we list detailed statistics of these generated graphs in 
\cref{tab:dataset}.

\input{Figures-Tables/benchmark-algo-table}

\paragraph{Benchmark Algorithms and Heuristics} 
We summarize the algorithms we compare in \cref{tab:benchmark}.
For our streaming algorithms, we use \stk{} to denote the primal-dual based \cref{alg:stk-djm},
and \textsc{Stkb} to denote the $b$-matching based 
\cref{alg:kmatching-to-kdm}. 
We compare these against four of the offline algorithms that were determined to be the most practical (in terms of runtime and solution quality)
by Hanauer et al.~\cite{hanauer2022fast}.
These include the iterative Greedy (\git{}) and iterative Global Paths algorithms (\gpait{}), 
the NodeCentered algorithm (\nc{}), and the $k$-Edge Coloring algorithm (\kec{}) that we have described in 
\cref{sec:related}. 
For these four offline algorithms, we use the heuristics and post-processing steps recommended in \cite{hanauer2022fast}, which we list in \cref{tab:benchmark}. 
We refer to \cite{hanauer2022fast} for a detailed description of these heuristics. 
For our semi-streaming algorithms, we implement the three heuristics described in \cref{sec:heur}. 
We use \textsc{dp} to denote the dynamic programming heuristic for the \stk{} algorithm, and \textsc{cc} and \textsc{m} for the 
common color and merge heuristics, respectively, for the \textsc{Stkb} algorithm.

\subsection{Comparison of Streaming Algorithms} \label{subsec:stcom}

We first compare six variants of our streaming algorithms amongst themselves. 
For the primal-dual  approach, we include the standard \stk{} algorithm and the \stkdp{} heuristic. %
For the $ b$-matching-based approach, we have the CC (common color) and M (merge) heuristics in addition to the standard \stkbnh{} algorithm, for a total of four combinations. 

In \cref{fig:tamu-stream-summary}, we show the relative quality results on the \tsmall{} graphs for the streaming algorithms. We set $\varepsilon = 0.001$ and
tested with $k \in \lrc{2,4,8,16,32,64,96}$, but observed that beyond $k=32$, all the algorithms computed similar weights,
as at this point, the solutions likely contained nearly the entire graph. 
Hence, we only report results up to $k=32$.
For each graph, algorithm, and $k$ value combination, we conduct \emph{five} runs and record the mean runtime, memory usage, and solution weight. 
We calculate \emph{relative time} by taking the ratio of the mean runtime for each algorithm to the mean runtime 
of a baseline algorithm. 
\emph{Relative memory} and \emph{relative weight} are similarly computed. 
We choose \stk{} as the baseline algorithm for runtime and memory comparisons 
and \stkdp{} as the baseline for weight comparisons. 
We show geometric means of the relative weights computed by each algorithm across all graphs and $k$ value combinations as box plots in \cref{fig:tamu-stream-summary} (a). 
The relative time and relative memory metrics across increasing $k$ values are plotted in
\cref{fig:tamu-stream-summary} (b) and (c), respectively. 

The relative weight of \stkdp{} is always one, so we do not show it in the plot. In terms of median relative weight (the red line), \stk{} is the second best, and \stkbccm{} is the third best. Surprisingly, 
while the worst-case approximation guarantee of the primal-dual-based approach is weaker than the $ b$-matching-based approaches, it provides weights that are better than the latter in nearly all instances. For runtimes, we see that the fastest algorithm is \stk{}, while the slowest are \stkbnh{} and \stkbm{}. 
\stkbcc{} and \stkbccm{} both have similar runtimes and are faster than \stkbnh{} and \stkbm{}. 
The runtime of \stkdp{} is between \stk{} and \stkb{}. 
In terms of memory usage, \stk{} requires the least, while \stkdp{} requires roughly twice as much memory as \stk{}($1.76$ -- $1.86 \times$ across $k$). 
The other four $ b$-matching-based algorithms behave similarly to each other and are worse than both \stk{} and \stkdp{}. 

From this experiment, \emph{we conclude that among these six streaming algorithm variants, the best three are \stk{}, \stkdp{}, and \stkbccm{}}. 
Hence, all the remaining experiments will report results only for these three variants of the streaming algorithms.

\begin{figure}[t!]
    \centering
    \includegraphics[width=0.95\textwidth]{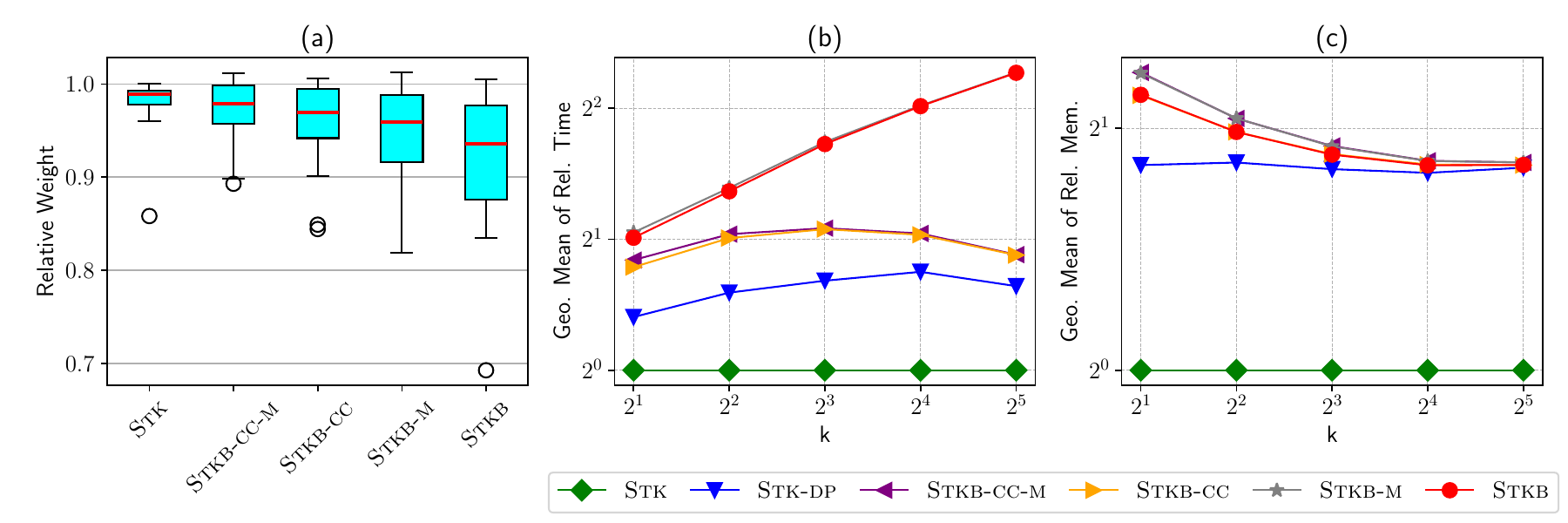}
    \caption{
        Summary plots for \tsmall{} instances on different streaming algorithms with $\varepsilon = 0.001$. 
        Plot (a) is a boxplot of relative weights across all instances and $k$ values for each algorithm. 
        Plots (b) and (c) give the geometric mean of the relative time and memory, respectively, across all instances with 
        increasing $k$ values. \stk{} is the baseline algorithm for relative time and memory, while \stkdp{} is the 
        baseline for relative weight.
    }
    \label{fig:tamu-stream-summary}
\end{figure}

\subsection{Comparison with Offline Algorithms}
\label{subsec:allcom}

Next, we compare the three streaming algorithms with the four offline algorithms listed in 
\cref{tab:benchmark}. We show the relative runtime, memory, and weight 
plots for the algorithms on the \tsmall{} dataset in \cref{fig:tamu-summary}. 
In \cref{subsec:experimental_results} we show the same metrics for the algorithms 
on the \gfive{} dataset in \cref{fig:g500-summary}. 
We follow the experimental settings and computations as in \cref{subsec:stcom} with 
\stk{} as the baseline for relative time and memory results, and \gpait{} with local swaps as the baseline for weight results, as these generally performed the best on their 
respective metrics.

We first discuss the \tsmall{} graph results. 
\emph{All of the streaming algorithms are significantly faster than the offline ones.} The fastest among these is the \stk{} algorithm, while the slowest is the $b$-matching based \stkb{}. Among the offline algorithms, \gpait{} is the slowest, more than $20 \times$ slower 
in geometric mean than \stk{}, while \git{} is more than $15 \times$ slower. The other two algorithms are relatively faster with similar 
runtimes but still slower than all streaming algorithms. The speedup for \stk{} w.r.t to \nc{} and \kec{} ranges from 3 to 11 across $k$. As an  example, for $k=8$, both \nc{} and \kec{} are more than $6\times$ slower 
than \stk{}. We also observe that both \nc{} and \kec{} get relatively more efficient as $k$ increases, which was also reported in \cite{hanauer2022fast}. 
For the memory results, we see that \stk{} requires the least, while the other two streaming algorithms take almost twice the memory, on average. 
All the offline algorithms behave similarly in terms of memory consumption since they all need to store the entire graph, which dominates the total memory consumption.  
\emph{We see a substantial memory reduction when using the streaming algorithms, with improvement ranging from 114$\times$ to 11$\times$ in geometric mean across $k$.} 
For smaller values of $k$ this reduction is more pronounced.
 
We now focus on the case $k=8$. All the streaming algorithms consume at least 16$ \times$ less memory than the offline algorithms, while for \stk{} it is $32 \times$. For the largest graph (kron\_g500-logn21) in this set, we see all the offline algorithms require at least 45 GB of memory while the streaming algorithms consume less than 1GB of memory. \emph{We emphasize that the higher memory requirement of the offline algorithms prohibits them from being run on larger datasets}, as we will see later. While the streaming algorithms are very efficient in terms of memory and time, we also see they obtain reasonably high solution weights.  %
For the weight results, we set \gpait{} as the baseline algorithm; hence, we do not include it in the box plot. 
All the offline algorithms find heavier weights than the streaming algorithms; for the \nc{} and \kec{} algorithms, 
we see many outliers compared to the other algorithms. \emph{ Among  the streaming algorithms,  \stkdp{} obtains the heaviest weight, with only less than 4\% median deviation from the best weight.} 
For \stkdp{}, the geometric mean of relative weights is 0.96 at $k=2$ and improves to 0.97 at $k=32$. The corresponding 
geometric mean of relative weights for faster offline algorithms, \nc{} and \kec{}  are as follows: for $k=2$, the means are 0.96 and 
0.97, respectively, and for $k=32$, they are 0.97 and 0.98, respectively. This highlights \stkdp{}'s comparable quality to the 
closest practical offline alternatives. The \stk{} and \stkb{} algorithms compute weights where the median deviation from the 
best weight is less than 5\% and 6\%, respectively.  

\begin{figure}[t!]
    \centering
    \includegraphics[width=0.95\textwidth]{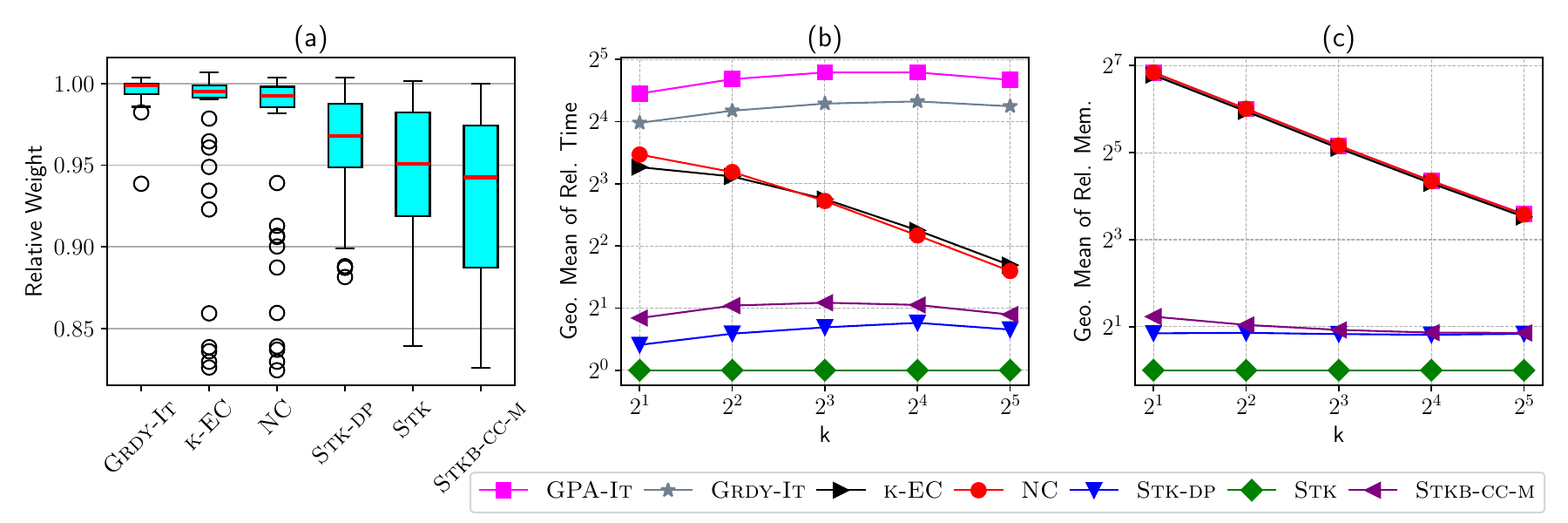}
    \caption{
        Summary plots the streaming and offline algorithms on \tsmall{} dataset with $\varepsilon = 0.001$ for the streaming algorithms. 
        Plot (a) is a boxplot of relative weights across all instances and $k$ values for each algorithm. 
        Plots (b) and (c) give the geometric mean of the relative time and memory, respectively, across all instances with 
        increasing $k$ values. \stk{} is the baseline algorithm for relative time and memory, while \gpait{} is the 
        baseline for relative weight.
    }
    \label{fig:tamu-summary}
\end{figure}

In \cref{subsec:experimental_results}, we include the results for similar experiments on the \gfive{} dataset in \cref{fig:g500-summary}. 
\emph{Overall, a similar conclusion can be drawn as the \tsmall{} instances.} The random graphs generated are much smaller than the \tsmall{} 
instances, and hence the memory improvements obtained by the streaming algorithms are smaller (6$\times$ to 38$\times$ in geometric mean). 
\emph{For the \gfive{} instances, the streaming algorithms obtain better quality results than the \tsmall{} instances}. The difference 
between the streaming and the \nc{} and \kec{} algorithms is \emph{smaller} than seen in the \tsmall{} instances. Both \nc{} and \stkdp{} 
achieve similar relative weights, while \kec{} is marginally (within 1\%) better.

\subsection{\tlarge{} Graph Results}

\input{Figures-Tables/stream-large-graphs}

We now discuss our \tlarge{} graph experiments. %
Since these graphs require longer runtimes, and our experiments on the smaller graphs reveal little deviation 
in runtime and memory across runs (the weight remains constant as our algorithms are deterministic), 
we report in \cref{tab:st-large} the results of a \emph{single run} of our streaming algorithms. 
We chose $k=8$ and set $\varepsilon = 0.001$ for this experiment. The first three columns represent the time in seconds, weight, and memory in GB for 
the baseline \stk{} algorithm, while the next six columns represent the relative metrics for the \stkdp{} and \stkb{} algorithms. 
For all the instances, using \stkdp{} yields an increase in solution quality over \stk{}, with the average increase being 2.93\%. 
Consistent with the results on smaller graphs, \stkb{} obtains the lowest weight among the streaming algorithms with weight decreasing 
in almost all the instances compared to \stk{} and the average decrease is 2.95\%. In terms of memory and runtime, 
\stkdp{} and \stkb{} require at most twice as much memory and time as the \stk{} algorithm. The geometric mean of relative memory and runtime of \stkdp{} is 1.85 and 1.78, respectively, and for \stkb{} they are 1.45 and 1.30, respectively.

For the offline algorithms, we chose \nc{} and \kec{},  since the previous experiments show they have much lower runtimes than the other two iterative matching algorithms. 
\emph{These algorithms could only be run on the smallest graph in this dataset (mycielskian20) while respecting the $1$ TB memory limit.} 
For this graph, \kec{} and \nc{} obtained weights of 1.70e+12 and 1.68e+12, respectively, which are around $18\%$ less than \stkdp{}. 
The \kec{} algorithm required more than two hours to compute a solution, while \nc{} required about twenty minutes. 
This is much worse than any of the streaming algorithms, as even the slowest one (\stkdp{}) required 
less than four minutes. 
Both the \nc{} and \kec{} algorithms used around 640 GB of memory, while the memory usage of the streaming 
algorithms ranges from $660$ MB for \stk{} to $1.4$ GB for \stkdp{}, which provides at least a $450$-fold reduction.

\paragraph{Effect of varying $\varepsilon$} 
In \cref{subsec:experimental_results}, we show experimental results highlighting 
the effects of varying $\varepsilon$ on the \tlarge{} graphs for the \stkdp{} algorithm in \cref{fig:eps-effect}.
The $\varepsilon$ parameter influences both the memory consumption and weight of the solution returned 
by the algorithm, and we find that as expected, increasing $\varepsilon$ decreases both of these values. However, in 
almost all cases, the decrease in weight is relatively much smaller than the decrease in memory, which suggests that 
using larger values of $\varepsilon$ in practice can substantially decrease the memory usage of the algorithm without significantly decreasing the weight of the solution returned.

\subsection{\tdata{} Graph Results}
\Cref{fig:dcn-fb-summary} shows experimental results for the graphs generated from the 
Facebook datacenter data. 
In \cref{subsec:experimental_results} we show the experimental results for 
the HPC and pFabric data in \cref{fig:dcn-hpc-summary,fig:dcn-pFabric-summary}, respectively. 
We use the same baseline algorithms and similar setup as the \tsmall{} dataset experiments. Overall the conclusion is similar to the earlier experiments, 
except that for these graphs, \stkb{} is the fastest. This is because the edge coloring step in the post-processing for the Facebook graphs is much faster 
than for the other graphs. For \stkdp{}, \stk{} and \stkb{} the median values of the geometric means of the relative weights are 0.96, 0.94, and 0.92, respectively. 
There are also substantial runtime and memory (10$\times$ -- 512$\times$) improvements compared to the offline algorithms.

\begin{figure}[t]
    \centering
    \includegraphics[width=0.95\textwidth]{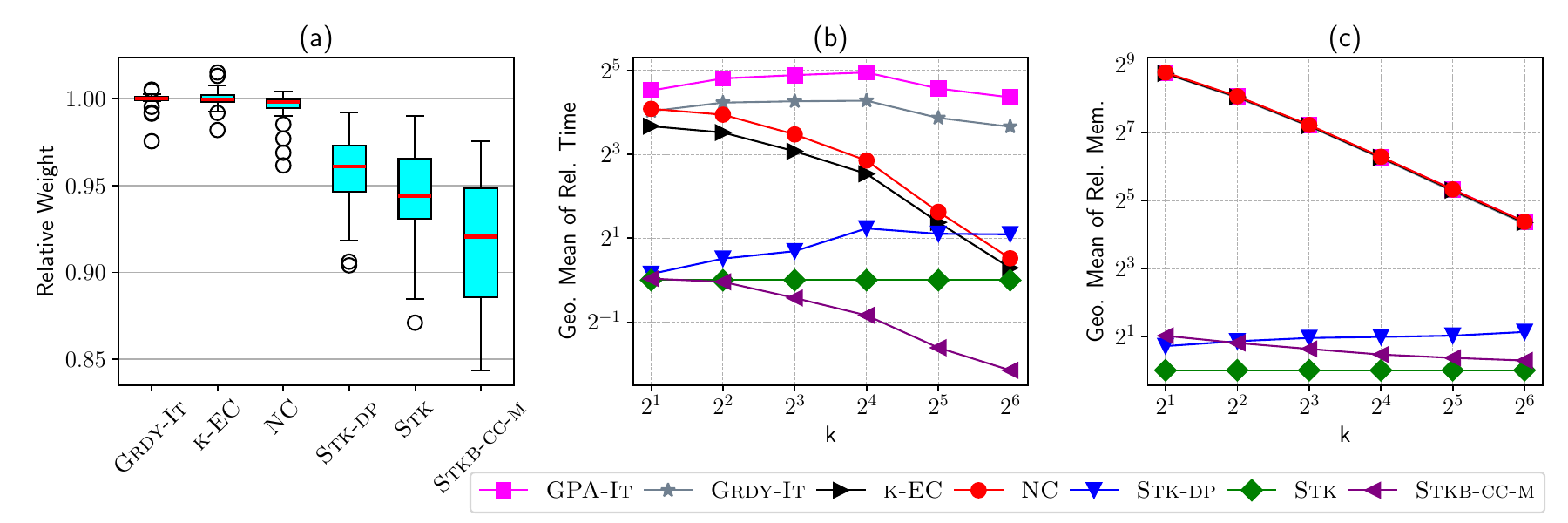}
    \caption{
        Summary plots the streaming and offline algorithms on Facebook \tdata{} dataset with $\varepsilon = 0.001$ for the streaming algorithms.
        Plot (a) is a boxplot of relative weights across all instances and $k$ values for each algorithm. 
        Plots (b) and (c) give the geometric mean of the relative time and memory, respectively, across all instances with 
        increasing $k$ values. \stk{} is the baseline algorithm for relative time and memory, while \gpait{} is the 
        baseline for relative weight.
    }
    \label{fig:dcn-fb-summary}
\end{figure}

%% file: Figures-Tables/benchmark-algo-table.tex
\begin{table}[t]
    \centering
    \footnotesize
    \begin{tabular}{@{}lllr@{}}\toprule
        Algorithm & Heuristics & Approx. & Time Complexity \\
        \midrule
        \git{}    & LS & $\nicefrac{1}{2}$ & $\bigO{\textrm{srt}(m)+km}$ \\
        \gpait{}  & LS & $\leq \nicefrac{1}{2}$ & $\bigO{\textrm{srt}(m)+km}$ \\
        \nc{}   & $\theta=0.2$, Agg=sum & $\leq \nicefrac{1}{2}$ & $\bigO{\textrm{srt}(n) + n \cdot \textrm{srt}(\Delta) + km}$ \\
        \kec{}     &   CC-RL    & $\leq \nicefrac{1}{2}$ & $\bigO{\textrm{srt}(m)+kn^2}$ \\
        \midrule
        \stk{} & \textsc{DP} & $\frac{1}{3+\epsilon}$ & $\bigO{km}$ \\
        \stkbnh{} & CC-M & $\frac{k}{(2+\varepsilon)(k+1)}$ & $\bigO{km + kn^2}$ \\
        \bottomrule
    \end{tabular}
    
    \caption{
        Benchmark approximation algorithms. LS: Local swaps, CC: Common color, RL: Rotate long, M: Merge, Agg: Aggregation, 
        $\text{srt}(x)$: Time complexity of sorting $x$ elements.  
    }
    \label{tab:benchmark}
\end{table}

%% file: Figures-Tables/stream-large-graphs.tex
\begin{table*}[t]
\small
\resizebox{0.95\textwidth}{!}{%
\begin{tabular}{@{}lrrrrrrrrr@{}}
     \toprule
     & \multicolumn{3}{c}{\stk{}} & \multicolumn{3}{c}{\stkdp{}} & \multicolumn{3}{c}{\stkb{}}\\ 
     \cmidrule(lr){2-4} \cmidrule(lr){5-7} \cmidrule(l){8-10}
     Graph &
     Time (s) &
     Weight &
     \begin{tabular}[c]{@{}r@{}}Mem.\\ (GB)\end{tabular} &
     \begin{tabular}[c]{@{}r@{}}Rel.\\ Time\end{tabular} &
     \begin{tabular}[c]{@{}r@{}}\% Wt. \\ Imprv.\end{tabular} &
     \begin{tabular}[c]{@{}r@{}}Rel. \\ Mem.\end{tabular} &
     \begin{tabular}[c]{@{}r@{}}Rel.\\ Time\end{tabular} &
     \begin{tabular}[c]{@{}r@{}}\% Wt. \\ Imprv.\end{tabular} &
     \begin{tabular}[c]{@{}r@{}}Rel.\\ Mem.\end{tabular} \\ 
     \midrule
     AGATHA\_2015   & 1377.54 & 1.60e+14 & 49.41 & 1.64 & 0.67 & 1.90 & 0.99   & -1.51   & 1.69  \\
     MOLIERE\_2016  & 736.75  & 8.26e+6  & 23.28 & 1.64 & 2.03 & 1.78 & 1.48   & 0.26    & 1.22  \\
     GAP-kron       & 629.37  & 1.20e+10 & 29.66 & 1.85 & 3.05 & 1.90 & 1.06   & -0.46   & 1.62  \\
     GAP-urand      & 679.73  & 9.83e+10 & 53.25 & 1.67 & 3.71 & 1.57 & 2.11   & -1.30   & 1.75  \\
     com-Friendster & 475.13  & 1.02e+14 & 22.66 & 1.62 & 2.84 & 1.81 & 1.49   & -4.58   & 1.54  \\
     mycielskian20  & 86.14   & 1.99e+12 & 0.65  & 2.34 & 5.30 & 2.17 & 0.98   & -10.10  & 1.03  \\ \bottomrule
\end{tabular}}
\caption{Comparison of streaming algorithms for $k=8$ and $\varepsilon = 0.001$ on \tlarge{} graphs.}
\label{tab:st-large}
\end{table*}

%% file: sections/conclusion.tex
\section{Conclusions and Future Work}
Earlier work on offline maximum weight matching algorithms showed that exact algorithms do not terminate on
graphs with hundreds of millions of edges. Hence, offline approximation algorithms with near-linear time complexities based on short augmentations were designed~\cite{pothen2019approximation}. 
However, our results show that on graphs with billions of edges, even these algorithms require over $1$ TB of memory for the \kdm{} problem, and do not terminate on such
graphs.

Streaming algorithms are designed to reduce  memory usage, and our streaming \kdm{} algorithms 
effectively reduce it by one to two orders of magnitude on our test set. 
Our results also show that the streaming algorithms are theoretically and empirically faster. 
In particular, \emph{we conclude that the \stkdp{} algorithm is the best performer since it only requires modestly more memory and runtime 
than the \stk{} algorithm while still computing solutions comparable (within 5\%) to the best offline algorithm}. 
Despite its weaker worst-case approximation ratio, we also find that \stk{} consistently outperforms \stkb{} in solution weight. %
This raises the question of whether the approximation ratio of \stk{} could be improved to $\frac{1}{2+\varepsilon}$.

%% file: sections/appendix_additional.tex
\section{Applications and Offline Algorithms}
\label{sec:kdm-app-offline}

\paragraph{General Applications} 
As stated before, unweighted \kdm{} has been previously studied in \cite{cockayne1978linear,feige2002approximating}.
Feige \etal\cite{feige2002approximating} motivate the problem by an application in scheduling connections in satellite networks. 
Cockayne \etal\cite{cockayne1978linear} motivate unweighted \textsf{$2$-DM} in bipartite graphs with a job assignment problem where a 
max cardinality 2-disjoint matching implies that on two successive days, maximum assignments can be scheduled such that no person performs the same 
job twice. This easily generalizes to a weighted setting where the goal is to schedule these disjoint assignments while maximizing some utility. 
As \kdm{} is closely related to \mwbm{} (with $b(v)=k$ for all $v \in V$), it is also potentially relevant to applications where $b$-matchings are used, such as graph 
construction in machine learning \cite{jebara2009graph}, load balancing in parallel environments \cite{ferdous2021parallel}, and privacy preservation in datasets \cite{khan2018adaptive}.

\paragraph{\kdm{} in Reconfigurable Datacenter Networks} 

Network traffic in datacenters is growing explosively due to their relevance to data science and machine learning. Since fixed topologies for networks that route data are oblivious to dynamically changing data traffic, 
reconfigurable optical technologies offer a promising alternative to existing static network designs. They augment static datacenter networks with reconfigurable optical matchings, where one edge-disjoint optical matching is used for each optical circuit switch. These optical matchings provide direct connectivity between racks, which allows for heavy traffic demands (elephant flows) to be routed through them. The remainder of the traffic demands (mice flows) are routed on static networks. Thus optical matchings can be adapted to meet dynamic traffic demands and can exploit temporal and spatial structure in the traffic data. The underlying optimization problem then becomes how these optical matchings that carry elephant flows can be computed quickly and efficiently. 
Matchings and $b$-matchings in offline, online, and dynamic settings (but not the semi-streaming setting) have been explored in this context in previous work; see among 
others \cite{ballani2020sirius,bienkowski2021online,bienkowski2022online,elhayek2023dynamic,hanauer2023dynamic,hanauer2022fast,mellette2017rotornet}.
In particular, for the offline \kdm{} problem, which models the scenario where there are $k$ optical switches and some central control plane has access to a traffic demand matrix, Hanauer \etal\cite{hanauer2022fast} provided several approximation algorithms and experimental results.

\paragraph{Offline Approximation Algorithms for \mwm{} and \mwbm{}}

A number of offline approximation algorithms have been designed in recent years for \mwm{} and \mwbm{},
many of which have also been implemented with codes available. Among these algorithms are Greedy, Locally Dominant, Path-growing, GPA, Suitor, etc.
Two surveys describing this extensive body of work are found in  
\cite{Hougardy:survey,pothen2019approximation}. 
These studies show that matching algorithms that employ short augmentations lead to  constant-factor 
approximation ratios (e.g., $\frac{1}{2}$, $(\frac{2}{3}-\varepsilon)$)
and near linear time complexities; practically they are fast and compute solutions with weights a few percentages off from being optimal,
and outperform more involved algorithms with better worst-case approximation ratios (e.g.,~$(1- \varepsilon$)) both in terms of time \emph{and} matching weight. 

%% file: sections/appendix_algorithms.tex
\section{Related Algorithms} \label{sec:additional_algs}

\subsection{Semi-Streaming Matching}
\label{subsec:alg_streamMatch}

\newcommand{\primalmwm}{\hyperlink{primal-mwm}{\textsf{(P-MWM)}}}
\newcommand{\dualmwm}{\hyperlink{dual-mwm}{\textsf{(D-MWM)}}}

The breakthrough $\frac{1}{2+\varepsilon}$-approximate semi-streaming algorithm (\ps{}) for maximum weighted matching is due to Paz 
and Schwartzman~\cite{PazS19}. The original algorithm was analyzed using local ratio techniques, but Ghaffari 
and Wajc~\cite{GhaffariW19} later provided a simpler primal-dual analysis of the algorithm which we adopt here. 
The primal-dual formulation of the \mwm{} problem is shown in \primalmwm{} and \dualmwm{}.

\input{sections/lp_MWM.tex}

\begin{algorithm}[ht!]
\caption{Semi-Streaming \mwm{} \cite{PazS19,GhaffariW19}}
\label{alg:stmatch}
\quad\;\, \textbf{Input:} Stream of edges $E$, a constant $\varepsilon > 0$ \\
\indent \quad\;\, \textbf{Output:} A matching $M$, using $\bigO{\frac{n\log^2{n}}{\varepsilon}}$ bits of space\\
\begin{minipage}[t]{0.6\linewidth}
\begin{algorithmic}[1]
\LComment{Initialization}
\State $\forall v\in V: \phi(v) \gets 0$
\State $S \leftarrow \emptyset$; $M \leftarrow \emptyset$
\Statex
\LComment{Streaming Phase}
\For{$e(u,v) \in E$} 
    \If{$w(e) \geq (1+\varepsilon) (\phi(u) + \phi(v))$} \label{algline:mwm_cond}
        \State{$w^\prime(e) \gets w(e) - (\phi(u) + \phi(v))$}
        \State{$\phi(u) \gets \phi(u) + w^\prime(e)$; $\phi(v) \gets \phi(v) + w^\prime(e)$}
        \State{$S$.push(e)}
    \EndIf
\EndFor
\end{algorithmic}
\end{minipage}
\hfill
\begin{minipage}[t]{0.38\linewidth}
\begin{algorithmic}[1]
\setcounter{ALG@line}{8}
\LComment{Post-Processing}
\While{$S$ is not empty} 
    \State{$e = (u,v) \leftarrow $ $S$.pop()}
    \If{$V(M) \cap \{u,v\} = \emptyset$} 
        \State{$M \leftarrow M \cup \{e\}$}
    \EndIf
\EndWhile
\end{algorithmic}
\end{minipage}
\end{algorithm}

The method is shown in \cref{alg:stmatch}. It initializes the approximate dual variables (the vector $\phi$) to zero, 
and then processes the streaming edges one by one. When an edge $e$ arrives, the algorithm decides whether to store it in the set of candidate matching edges (the stack $S$) 
or to discard it. This decision is based on whether the dual constraint (shown in line \ref{algline:mwm_cond}) is approximately satisfied for this edge. If the edge is 
stored, we compute the reduced weight $w^\prime(e) = w(e) - (\phi(u) + \phi(v))$ and add it to both $\phi(u)$ and $\phi(v)$. 
Ghaffari and Wajc \cite{GhaffariW19} showed that as edges incident on a vertex $v$  re inserted into the stack $S$, they have weights that exponentially increase with the factor $1+\varepsilon$. 
Thus, for each vertex at most $\bigO{\log_{1+\varepsilon} W} = \bigO{\frac{\log W}{\varepsilon}} = \bigO{\frac{\log n}{\varepsilon}}$ edges are stored in $S$ 
(since we assume $W = \frac{w_{\max}}{w_{\min}}$ to be $\poly{n}$ in this paper), which implies the stack has size $\bigO{\frac{n\log n}{\varepsilon}}$. 
In the post-processing phase, the algorithm unwinds the stack and greedily constructs a maximal matching by processing edges in the stack order, in serial.
Through a primal-dual based analysis, Ghaffari and Wajc \cite{GhaffariW19} show that the resulting matching has an approximation ratio of $\frac{1}{2 + \varepsilon}$ (up to appropriately 
scaling the $\varepsilon$ factor by a constant). 

This algorithm was implemented by Ferdous \etal\cite{Ferdous+:2024a} and it was shown to reduce the memory requirements by one to two factors order of magnitude over offline $\frac{1}{2}$-approximation algorithms, 
while being close to the best of them in run time and matching weight. The matching weight was a few percent off the weight obtained from the offline algorithm that is currently the best for weights, a $(\frac{2}{3} - \varepsilon)$-approximation algorithm. 
However, on one of the largest problems with several billions of edges, the $(\frac{2}{3} - \varepsilon)$-approximation algorithm did not terminate even when run on a shared memory parallel computer with 1 TB of memory, while the streaming algorithm used 
less than $6$ GB of memory on a serial machine.

\subsection{Semi-Streaming \texorpdfstring{$b$}{b}-Matching} \label{sec:alg_streambmatch}

Recall that given a function $b \colon V \to \Z_+$, a $b$-matching is a set of edges $F \sse E$ such that each vertex $v \in V$ is incident on at most $b(v)$ edges in $F$. 
At a high level, the semi-streaming \mwbm{} algorithm of Huang and Sellier \cite{huang2021semi} works by maintaining a global stack of edges that is then greedily unwound to construct a $b$-matching. In particular, 
for a chosen parameter $\varepsilon > 0$, 
the size of the stack is bounded by $\bigO{\abs{F_{\max}} \cdot \log_{1+\varepsilon} (W\varepsilon^{-1})}$ edges, where $\abs{F_{\max}}$ is the maximum size of any $b$-matching in the graph and $W = \frac{w_{\max}}{w_{\min}}$ is the ratio of the maximum and minimum edge weights.
Additionally, it can be shown that within the edges of the stack there exists a $\frac{1}{2+\varepsilon}$-approximate $b$-matching, which is retrieved by greedily unwinding the stack. Note that for our specific usage, we have that $b(v) = k$ for all $v \in V$ and that $W = \poly{n}$, which implies that the algorithm stores $\bigO{nk \log n}$ edges for any constant $\varepsilon > 0$, and hence requires $\bigO{nk \log^2 n}$ bits of space. The pseudocode of a slight modification of the original semi-streaming algorithm is given in \cref{alg:ssmwbm}, and for the formal details on proofs of the space complexity and approximation factor we refer to Huang and Sellier \cite{huang2021semi}. 

Informally, \cref{alg:ssmwbm} maintains for each vertex $v \in V$ and $i \in [b(v)]$ a value $\phi(v, i)$ and a pointer $t_v(i)$ which are initially set to 0 and $\emptyset$, respectively, and a stack of edges $\calS$. 
Each vertex $v$ contributes at most $b(v)$ edges in the final solution, and so we can keep track of the $i^{\text{th}}$ chosen edge with the pointer $t_v(i)$. The $\phi(v, i)$ value can be interpreted as the gain in solution weight for the edge stored in $t_v(i)$. 
For each edge $e = (u, v)$ that streams in, we find the indices and values $q_u$, $\phi(u) \coloneqq \phi(u, q_u)$ and $q_v$, $\phi(v) \coloneqq \phi(v, q_v)$ that give the minimum among $\phi(u, i)$, $i \in [b(u)]$, and $\phi(v, j)$, $j \in [b(v)]$, respectively. The edge $e$ is then pushed to $\calS$ only if it satisfies $w(e) \geq (1 + \nicefrac{\varepsilon}{2})(\phi(u) + \phi(v))$. Intuitively, this enforces that the edges incident on each vertex $v \in V$ that get pushed to the stack have exponentially increasing gain in solution weight. If the edge $e$ satisfies this condition, then it will become the $q_u^{\text{th}}$ and $q_v^{\text{th}}$ chosen edges of $u$ and $v$, respectively, so we update the pointers $t_u(q_u)$ and $t_v(q_v)$ to $e$. We also increase the values of $\phi(u, q_u)$ and $\phi(v, q_v)$ by the reduced weight of $e$, which is given by $w'(e) = w(e) - (\phi(u) + \phi(v))$. Additionally, we keep pointers $p_u(e)$ and $p_v(e)$ to the $q_u^{\text{th}}$ chosen edge of $u$ and the $q_v^{\text{th}}$ chosen edge of $v$ that $e$ just replaced. 

Once all the edges have streamed and been processed, the $b$-matching can be retrieved by greedily unwinding the stack $\calS$. Specifically, for each edge in the stack we maintain a boolean flag that is initially true, and an edge can only be added to the solution if this flag is true. Clearly, the first popped edge must be added to the solution. For an edge $e = (u, v)$ that is added to the solution, we do pointer chasing on the pointers $p_u(e)$ and $p_v(e)$ until they both point to $\emptyset$, and for each edge found in the chases, we set their boolean flag to false. Intuitively, this can be seen as ignoring all the edges that were pushed earlier in the stack that at some point were the $q_u^{\text{th}}$ and $q_v^{\text{th}}$ chosen edges of $u$ and $v$, respectively. 

\begin{algorithm}[t]
    \caption{Semi-Streaming \mwbm{} \cite{huang2021semi}}
    \label{alg:ssmwbm}

    \quad\;\, \textbf{Input:} A stream of edges $E$, a function $b \colon V \to \Z_+$, and a constant $\eps > 0$ \\
    \indent \quad\;\, \textbf{Output:} A $\frac{1}{2+\varepsilon}$-approximate $b$-matching $F$ \\
    \begin{minipage}[t]{0.6\linewidth}
    \begin{algorithmic}[1]
        \LComment{Initialization}
        \State{$\calS \gets \emptyset$} \Comment{Global stack of edges}
        \For{$v \in V$} 
            \For{$i \in [b(v)]$}
                \State $\phi(v, i) \gets 0$, $t_v(i) \gets \emptyset$
            \EndFor
        \EndFor
        \Statex
        \LComment{Streaming Phase}
        \For{$e = (u, v) \in E$}
            \State $q_u \gets \argmin_{q \in [b(u)]} \lrc{\phi(u, q)}$, $\phi(u) \gets \phi(u, q_u)$
            \State $q_v \gets \argmin_{q \in [b(v)]} \lrc{\phi(v, q)}$, $\phi(v) \gets \phi(v, q_v)$

            \If{$w(e) \geq (1+\nicefrac{\varepsilon}{2})(\phi(u) + \phi(v))$}
                \State $w^\prime(e) \gets w(e) - (\phi(u) + \phi(v))$ \Comment{Gain of $e$}
                \State $\calS.\text{push}(e)$

                \State $p_u(e) \gets t_u(q_u)$, $p_v(e) \gets t_v(q_v)$ %
                \State $\phi(u, q_u) \gets \phi(u) + w'(e)$, $\phi(v, q_v) \gets \phi(v) + w'(e)$ %
                \State $t_u(q_u) \gets e$, $t_v(q_v) \gets e$ \Comment{Update pointers}
            \EndIf
        \EndFor
    \end{algorithmic}
    \end{minipage}
    \hfill
    \begin{minipage}[t]{0.35\linewidth}
    \begin{algorithmic}[1]
    \setcounter{ALG@line}{15}
        \LComment{Post-Processing}
        \State{$F \gets \emptyset$}
        \State{$\forall e \in \calS \colon a_e \gets $true}
        \While{$\calS \neq \emptyset$}
            \State $e = (u, v) \gets \calS.\text{pop()}$
            \If{$a_e = $ true}
                \State $F \gets F \cup \lrc{e}$
                \For{$x \in \lrc{u, v}$}
                    \State $c \gets e$
                    \While{$c \neq \emptyset$}
                        \State $a_c \gets$ false
                        \State $c \gets p_x(c)$
                    \EndWhile
                \EndFor
            \EndIf
        \EndWhile
        \State{\Return $F$}
    \end{algorithmic}
    \end{minipage}
\end{algorithm}

\subsection{(\texorpdfstring{$\Delta+1$}{Delta+1})-Edge Coloring} \label{sec:alg_misragries}

The edge coloring algorithm of Misra and Gries \cite{misra1992constructive} constructs a proper edge coloring $\calC \colon E \to [\Delta +1]$, i.e., a coloring such that for any two adjacent edges do not have the same color. In particular, it does so in $\bigO{nm}$ time using $\bigO{m}$ space. At a high level, the main procedure of the algorithm colors an uncolored edge while maintaining the invariants that colored edges will never become uncolored (but may change colors) and that if all colored edges at the start of the procedure form a proper coloring, then the resulting colored edges will also be proper. By a simple induction, the final coloring must therefore be proper if this procedure is iteratively applied on uncolored edges. The pseudocode of the algorithm is given in \cref{alg:edgecoloring}.
For technical details on the proof that it constructs a proper $(\Delta + 1)$-edge coloring, we refer to the original paper of Misra and Gries \cite{misra1992constructive}.

\begin{algorithm}[t]
    \caption{$(\Delta + 1)-$Edge Coloring \cite{misra1992constructive}}
    \label{alg:edgecoloring}

    \quad\;\, \textbf{Input:} A graph $G = (V, E)$ \\
    \indent \quad\;\, \textbf{Output:} A proper edge coloring $\calC \colon E \to \lrb{\Delta + 1}$ \\
    \begin{minipage}[t]{0.56\linewidth}
    \begin{algorithmic}[1]
        \State{$\Delta \gets \max_{v \in V}{\deg{v}}$} 
        \State{$\forall e \in E \colon \calC(e) \gets \perp$} \Comment{Each edge starts uncolored}
        \For{$e = (u, v) \in E$} 
            \If{$\calC(e) = \perp$} 
                \LComment{Common Color Heuristic}
                \State $\ell \gets$ color in $\lrb{\Delta+1} \cup \lrc{\perp}$ that is free on $u$ and $v$
                \If{$\ell \neq \perp$}
                    \State $\calC(e) \gets \ell$
                    \State \textbf{continue}
                \EndIf
                \State
                \State $F \gets \Call{MaximalFan}{u,\, v}$, $z \gets F.\text{back()}$
                \State $c \gets$ a color in $\lrb{\Delta+1}$ that is free on $u$
                \State $d \gets$ a color in $\lrb{\Delta+1}$ that is free on $z$
                \If{$d$ is not free on $u$}
                    \State \Call{CDPath}{$u$, $c$, $d$}
                    \State $F \gets$ \Call{ShrinkFan}{$F$, $d$}, $z \gets F.\text{back()}$
                \EndIf
                \State \Call{RotateFan}{$F$}, $e' \gets (u, z)$
                \State $\calC(e') \gets d$
            \EndIf
        \EndFor
        \State \Return $\calC$

        \Statex
        \Procedure{MaximalFan}{$u$, $v$}
            \State $F \gets \lrv{v}$ \Comment{Array}
            \State $T \gets \lrc{x \in N(u) \colon \calC(u, x) \neq \perp}$
            \While{$\exists x \in T$ such that $\calC(u, x)$ is free on $F$.back()}
                \State $F.\text{push}(x)$
                \State $T \gets T \sm \lrc{x}$
            \EndWhile
            \State \Return $F$
        \EndProcedure
    \end{algorithmic}
    \end{minipage}
    \hfill
    \begin{minipage}[t]{0.4\linewidth}
    \begin{algorithmic}[1]
    \setcounter{ALG@line}{26}
            \Procedure{CDPath}{$u$, $c$, $d$}

            \State $x \gets u$, $z \gets \emptyset$
            \State $q \gets d$
            \While{$q$ is not free on $x$}
                \State $y \gets$ vertex in $N(x) \sm z$ such that $\calC(x, y) = q$  
                \State $p \gets c$ \textbf{if} $q = d$ \textbf{else} $d$ 
                \State $\calC(x, y) \gets p$
                \State $z \gets x$, $x \gets y$, $q \gets p$
            \EndWhile
        \EndProcedure

        \Statex
        \Procedure{ShrinkFan}{$F$, $d$}
            \State $l \gets \abs{F}$
            \For{$i \in [l]$}
                \If{$d$ is free on $F[i]$} 
                    \State $F \gets \lrv{F[1], \ldots, F[i]}$
                    \State \textbf{break}
                \EndIf
            \EndFor 
            \State \Return $F$
        \EndProcedure
        
        \Statex
        \Procedure{RotateFan}{$F$}
            \State $l \gets \abs{F}$
            \For{$i \in [l-1]$}
                \State $x \gets F[i]$, $e_1 \gets (u, x)$
                \State $y \gets F[i+1]$, $e_2 \gets (u, y)$
                \State $\calC(e_1) \gets \calC(e_2)$
            \EndFor
            \State $z \gets F[l]$, $e_3 \gets (u, z)$
            \State $\calC(e_3) \gets \perp$
        \EndProcedure
    \end{algorithmic}
    \end{minipage}
\end{algorithm}

For a vertex $v \in V$, we say that a color $c \in [\Delta + 1]$ is \emph{free} on $v$ if there exists no edge $e'$ incident on $v$ such that $\calC(e') = c$. 
The main loop of \cref{alg:edgecoloring} for coloring an uncolored edge $e = (u, v)$ works as follows. As a heuristic, we can first check if there exists some color $\ell \in [\Delta + 1]$ that is free on both $u$ and $v$, and if so set $\calC(e) = \ell$ and continue to the next edge. Otherwise, we construct a data structure $F$ called a maximal fan. The fan $F$ is a maximal ordered list of neighbors of $u$ such that $F[1] = v$ and for $2 \leq i \leq l = \abs{F}$, the color $\calC(u, F[i]) \neq \perp$ and is free on $F[i-1]$. Once $F$ is constructed, we denote by $z$ the last vertex in $F$. We then find a color $c \in [\Delta+1]$ that is free on $u$ and a color $d \in [\Delta+1]$ that is free on $z$. Since $\degree{u}, \degree{z} \leq \Delta$, such colors must always exist. 

If $d$ is free on $u$, then we perform a rotation of the fan $F$, which involves circularly shifting the colors of the corresponding edges of the fan by one to the left, i.e., setting $\calC(u, F[i]) = \calC(u, F[i+1])$ for $1 \leq i \leq l-1$ and setting $\calC(u, z) = \perp$. 
Since we have that $d$ is free on both $u$ and $z$, we can safely set $\calC(u, z) = d$ and finish. Otherwise, if $d$ is not free on $u$, we first find a structure called a $cd_u$ path, which is simply a maximal path of edges starting at $u$ such that the colors of edges on the path alternate between $c$ and $d$. Note that since $c$ is free on $u$, this implies that the first edge $(u, x)$ on such a path must have color $d$ and additionally that $x$ must be in the fan before $z$ (as otherwise we could have increased the size of the fan by adding $x$ after $z$). Once such a path is found, we simply invert the colors of the edges on the path, i.e., set all edges with color $c$ to $d$ and vice versa. This operation now ensures that $d$ is free on $u$, but does not guarantee that $d$ remains free on $z$. To fix this, we can shrink the fan $F$ up to the first vertex in $F$ such that $d$ is free on it and update $z$ accordingly. We can safely do a rotation of the shrunken fan and set $\calC(u, z) = d$.

%% file: sections/lp_MWM.tex
\begin{figure}[b]
    \begin{minipage}[t]{.45\columnwidth}
    \begin{maxi*}
        {}{\sum_{e \in E} w(e) x(e) }{}{\hypertarget{primal-mwm}{\textsf{(P-MWM)}}}{} %
        \addConstraint{\sum_{e \in \delta(v)} x(e)}{\leq 1}{\; \forall v \in V}
        \addConstraint{ x(e)}{\geq 0}{\; \forall e \in E}.
    \end{maxi*}
    \end{minipage}%
    \begin{minipage}[t]{.55\columnwidth}
    \begin{mini*}
        {}{\sum_{v \in V} y(v)}{}{\hypertarget{dual-mwm}{\textsf{(D-MWM)}}}{} %
        \addConstraint{y(u) + y(v) }{\geq w(e)}{\; \forall e=(u, v) \in E}
        \addConstraint{y(v)}{\geq 0}{\; \forall v \in V}.
    \end{mini*}
    \end{minipage}
    \caption{LP Relaxation \protect\primalmwm{} of \mwm{} and its dual \protect\dualmwm{}.}
    \label{fig:lp_mwm}
\end{figure}

%% file: sections/appendix_experiments.tex
\section{Additional Experimental Details}
\label{sec:appendix_expr}

\subsection{Dataset Description} \label{subsec:dataset}

\cref{tab:dataset,tab:rmat_graph_stats} include sizes and degree measures for the 
95 graphs that we have reported results on.

\input{Figures-Tables/Datasets}

\input{Figures-Tables/rmat_graph_stats.tex}

\subsection{Other Experimental Results} \label{subsec:experimental_results}

\cref{fig:g500-summary} shows the summary results for the \gfive{} graphs and 
\cref{fig:eps-effect} shows the relative weight and relative memory 
results for the \stkdp{} algorithm on the \tlarge{} dataset when using varying values of $\varepsilon$. 
The summary of experimental results for HPC and pFabric datacenter network \tdata{} graphs are shown in \cref{fig:dcn-hpc-summary,fig:dcn-pFabric-summary}, respectively.

\begin{figure}[!ht]
    \centering
    \includegraphics[width=0.9\linewidth]{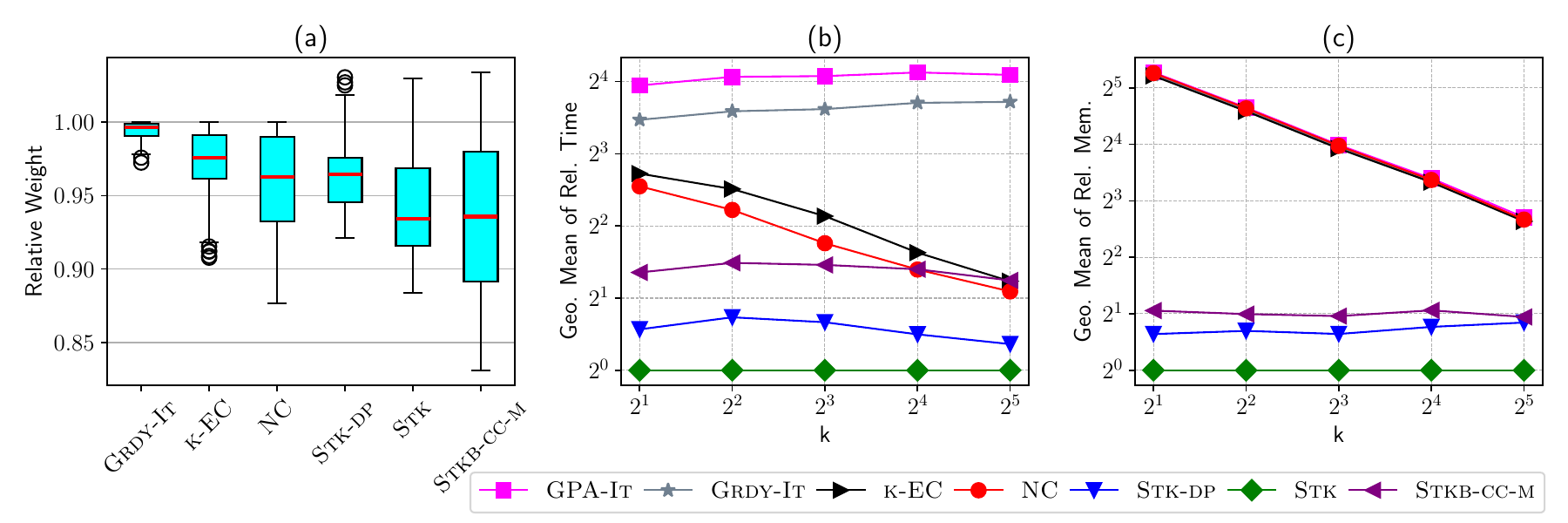}
    \caption{Summary plots for the streaming and offline algorithms on \gfive{} graphs. (a) Boxplot of relative weights across all instances and $k$ values for each algorithm. We set $\varepsilon = 0.001$ for \stk{}, \stkdp{} and \stkb{}. (b) Geometric mean of relative time and (c) geometric mean of relative memory,  across all instances with increasing $k$ values.  \gpait{} is the baseline for relative weight, and \stk{} is the baseline algorithm for relative time and memory.
Note the logarithmic scales in the axes of the last two subplots.}
    \label{fig:g500-summary}
\end{figure}

\begin{figure}
    \centering
    \includegraphics[width=0.9\linewidth]{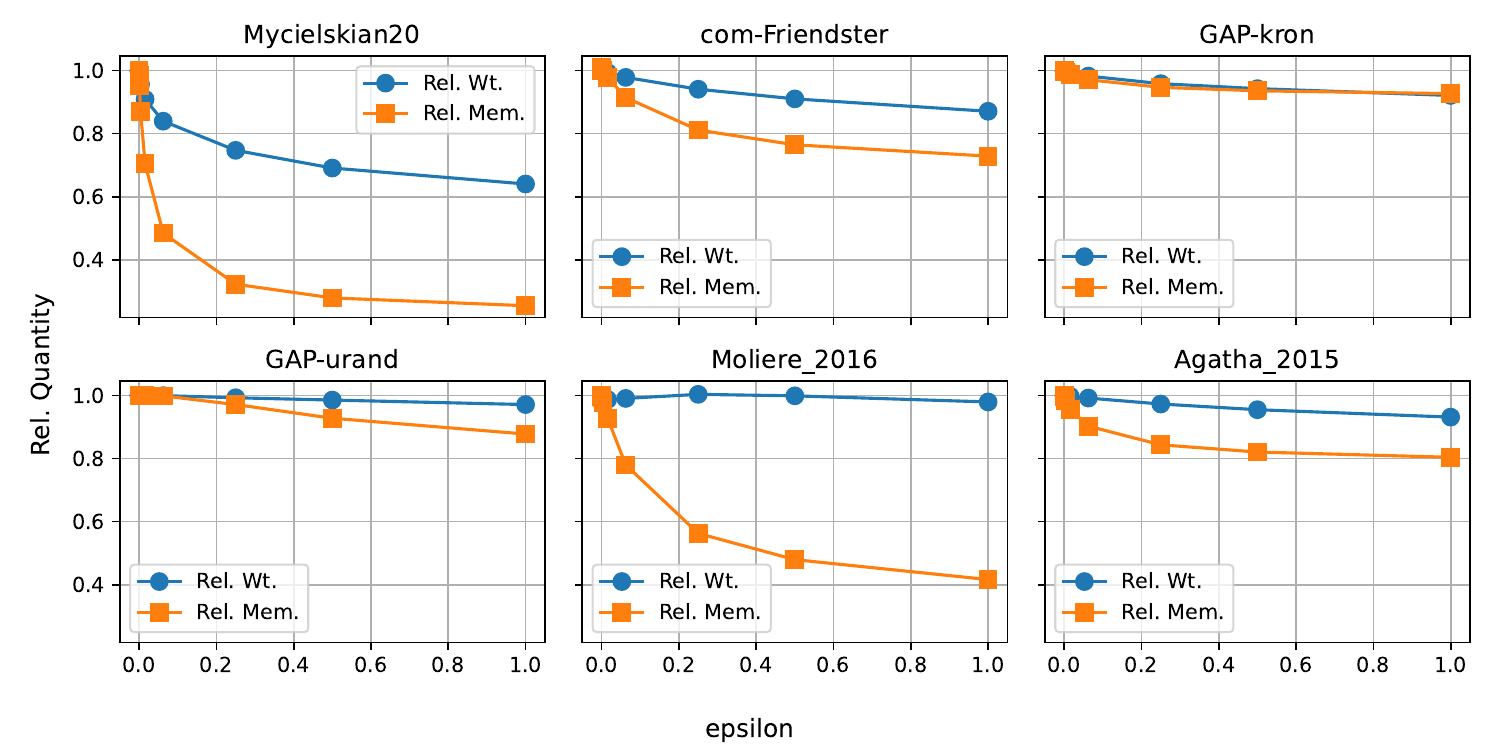}
    \caption{Relative weights and memory of the \tlarge{} graphs with varying $\varepsilon$ using \stkdp{}, with $k=8$. 
    The quantities are relative to $\varepsilon = 0$ values, and we test with $\varepsilon \in \{0,2^{-x}\}$, where $x \in \{16,14,12,10,8,6,4,2,1,0\}$.}
    \label{fig:eps-effect}
\end{figure}

\begin{figure}[t]
    \centering
    \includegraphics[width=0.95\textwidth]{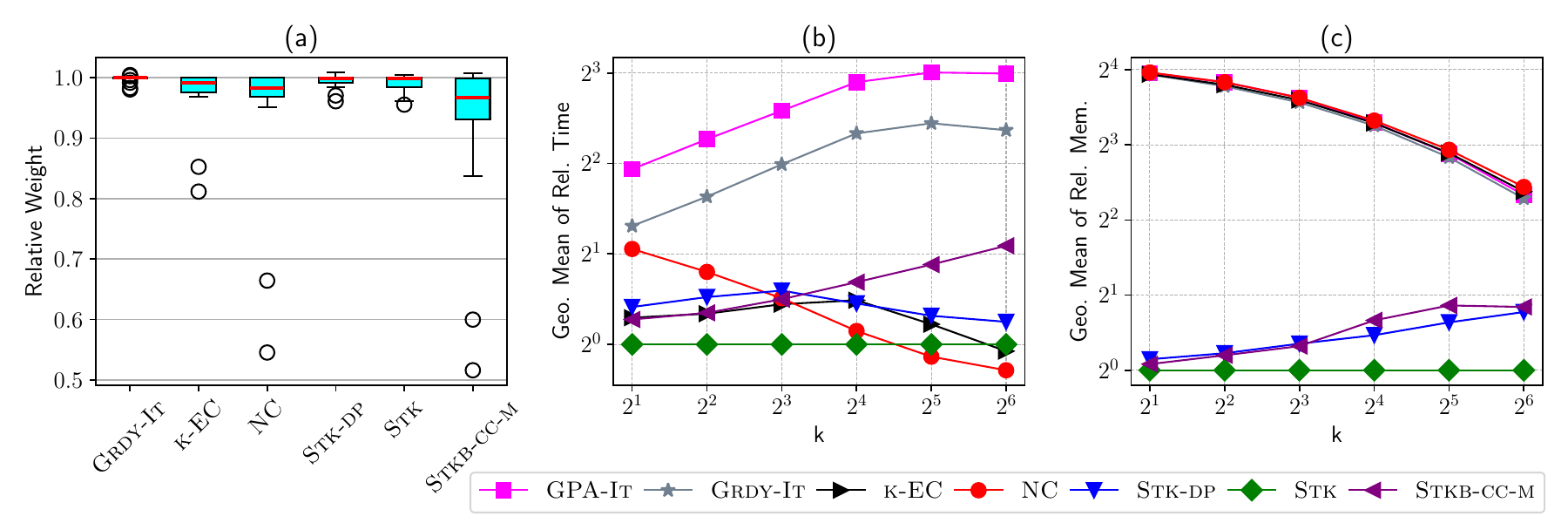}
    \caption{Summary plots the streaming and offline algorithms on HPC \tdata{} dataset. (a) Boxplot of relative weights across all instances and $k$ values for each algorithm, (b) Geometric mean of relative time, and (c) geometric mean of relative memory across all instances with increasing $k$ values.  We set $\varepsilon = 0.001$ for \stk{}, \stkdp{} and \stkb{}.
    \gpait{} is the baseline for relative weight, and \stk{} is the baseline algorithm for relative time and memory.
    Note the logarithmic scales in the axes of the last two subplots.}
    \label{fig:dcn-hpc-summary}
\end{figure}

\begin{figure}[t]
    \centering
    \includegraphics[width=0.95\textwidth]{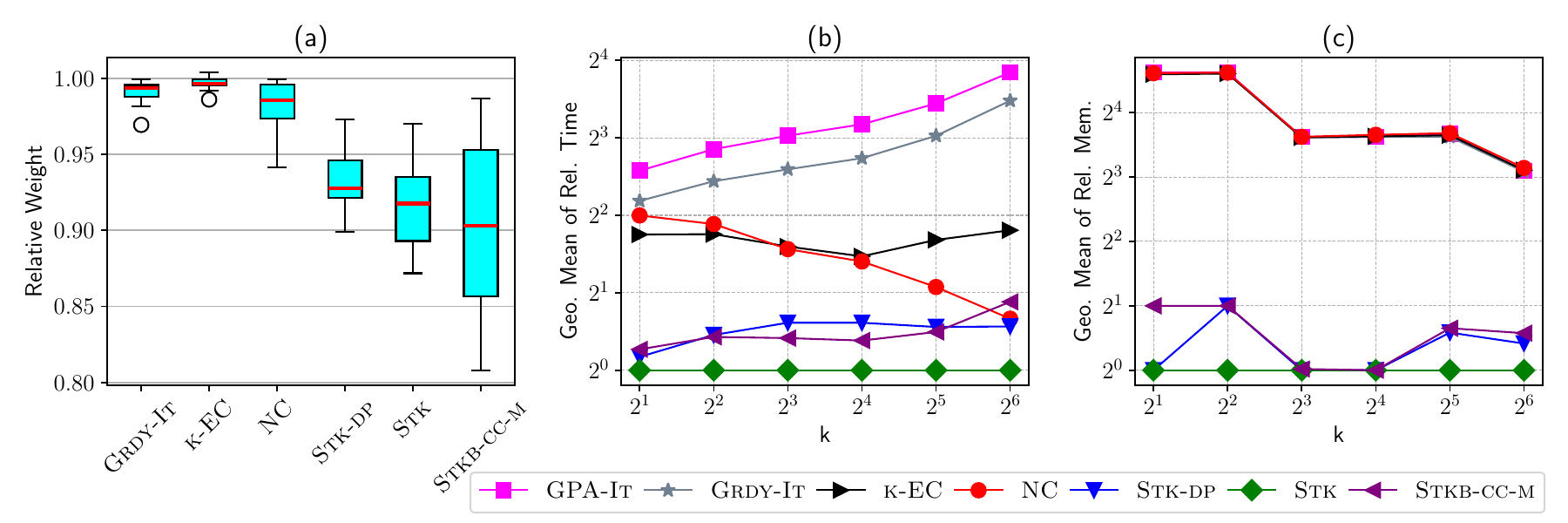}
    \caption{Summary plots the streaming and offline algorithms on pFabric \tdata{} dataset. (a) Boxplot of relative weights across all instances and $k$ values for each algorithm, (b) Geometric mean of relative time, and (c) geometric mean of relative memory across all instances with increasing $k$ values.  We set $\varepsilon = 0.001$ for \stk{}, \stkdp{} and \stkb{}. \gpait{} is the baseline for relative weight, and \stk{} is the baseline algorithm for relative time and memory. Note the logarithmic scales in the axes of the last two subplots.}
    \label{fig:dcn-pFabric-summary}
\end{figure}

%% file: Figures-Tables/Datasets.tex
\begin{table}[!ht]
\centering

\begin{subtable}[t]{\columnwidth}
    \centering
        \begin{tabular}{@{}l|rrrrr@{}}\toprule
            Graph &$n$ &$m$ &Avg. Deg. &Max. Deg.  & Min. Deg.\\\midrule
            mycielskian20 (U) &786.43 K &1.36 B &3446.42 &393,215 & 19 \\
            com-Friendster (U) &65.61 M &1.81 B &55.06 &5,214 & 1 \\
            GAP-kron (W) &134.22 M &2.11 B &31.47 &1,572,838 & 0  \\
            GAP-urand (W) &134.22 M &2.15 B &32 &68 & 6  \\
            MOLIERE\_2016 (W) &30.24 M &3.34 B &220.81 &2,106,904 & 0 \\
            Agatha\_2015 (U) &183.96 M &5.79 B &62.99 & 12,642,631 & 1 \\
            \bottomrule
        \end{tabular}
    \caption{\tlarge{} graph instances. U: Unweighted, W: Weighted Graph. K: Thousand, M: Million, B: Billion.}
    \label{tab:tlarge-data}
\end{subtable}%
\vspace{1em}
\begin{subtable}[t]{\columnwidth}
    \centering
        \begin{tabular}{@{}l|rrrrr@{}}
            \toprule
            Graph	& $n$ &	$m$ & Avg. Deg. & Max. Deg. & Min. Deg. \\
            \midrule
            astro-ph & 16,706 & 121,251 & 14.52 & 360 & 0 \\
            Reuters911 & 13,332 & 148,038 & 22.21 & 2,265 & 0\\
            cond-mat-2005 & 40,421 & 175,691 & 8.69 & 278 & 0 \\
            gas\_sensor & 66,917 & 818,224 & 24.45 & 32 & 7 \\
            turon\_m & 189,924 & 778,531 & 8.20 & 10 & 1 \\
            Fault\_639 & 638,802 & 13,303,571 & 41.65 & 266 & 0 \\
            mouse\_ gene & 45,101 & 14,461,095 & 641.27 & 8,031 & 0 \\
            bone010 & 986,703 & 23,432,540 & 47.50 & 62 & 11 \\
            dielFil.V3real & 1,102,824 & 44,101,598 & 79.98 & 269 & 8\\
            kron.logn21 & 2,097,152 & 91,040,932 & 86.82 & 213,904 & 0 \\
            \bottomrule
        \end{tabular}
    \caption{\tsmall{} graph instances. All the graphs are weighted.}
    \label{tab:tsmall-data}
\end{subtable}
\begin{subtable}[t]{\columnwidth}
    \centering
        \begin{tabular}{@{}l|rrrrrrrr@{}}\toprule
        Graph &$n$ &$m$ &Avg. &Max. &Min. &Max. &Min. &\#Trcs \\
& & &Deg. &Deg. &Deg. &Dem. &Dem. &(M) \\\toprule
FB\_ClusterA\_rack &13,733 &496,624 &72.33 &7,272 &1 &142,053 &1 &316 \\
FB\_ClusterB\_rack &18,897 &1,777,559 &188.13 &11,932 &1 &315,718 &1 &2,710 \\
FB\_ClusterC\_rack &27,358 &2,326,086 &170.05 &25,224 &1 &169,202 &1 &302 \\
FB\_ClusterA\_ip &357,059 &43,057,511 &241.18 &57,676 &1 &18,614 &1 &316 \\
FB\_ClusterB\_ip &4,963,141 &164,277,914 &66.20 &376,508 &1 &17,036 &1 &2,710 \\
FB\_ClusterC\_ip &990,023 &40,654,711 &82.13 &104,821 &1 &18,571 &1 &316 \\
\midrule
HPC1 &1,024 &3,797 &7.42 &21 &0 &1,060 &530 &3 \\
HPC2 &1,024 &15,095 &29.48 &36 &0 &2,071 &3 &22 \\
HPC3 &1,024 &37,908 &74.04 &1,022 &0 &48 &2 &1 \\
HPC4 &1,024 &10,603 &20.71 &26 &0 &1,690 &1690 &18 \\
\midrule
pFabric\_0.1 &144 &10,296 &143.00 &143 &143 &41,832 &1 &30 \\
pFabric\_0.5 &144 &10,296 &143.00 &143 &143 &43,897 &14 &30 \\
pFabric\_0.8 &144 &10,296 &143.00 &143 &143 &39,208 &14 &30 \\
\bottomrule
\end{tabular}
    \caption{\tdata{} instances. The edge weights are the number of occurrences of a pair of vertices in the trace data. Self loops are discarded. 
    Max.~Dem. and Min.~Dem. are the maximum and minimum weights on the edges representing the demands, while \#Trcs list the number of traces in millions.}
    \label{tab:trace-data}
\end{subtable}%

\caption{Graph statistics for \tlarge{}, \tsmall{} and \tdata{} instances.} \label{tab:dataset}
\end{table}

%% file: Figures-Tables/rmat_graph_stats.tex
\begin{table*}[!ht]
    \centering
    \footnotesize
    \begin{tabular}{@{}lrrrrrrrrr@{}}
    \toprule
    & \multicolumn{3}{c}{$\mathsf{rmat_{\text{b}}}$}             & \multicolumn{3}{c}{$\mathsf{rmat_{\text{g}}}$}             & \multicolumn{3}{c}{$\mathsf{rmat_{\text{er}}}$}            \\ 
    \cmidrule(lr){2-4} \cmidrule(lr){5-7} \cmidrule(l){8-10} 
    $\log n$ & $m$       & Avg. Deg.   & $\Delta$    & $m$       & Avg Deg.   & $\Delta$ & $m$ & Avg  Deg. & $\Delta$ \\ \midrule
    10       & 7,939     & 15.51 & 372    & 7,960     & 15.55 & 85  & 8,185     & 15.99 & 28  \\
    11       & 16,025    & 15.65 & 615    & 16,030    & 15.65 & 120 & 16,377    & 15.99 & 31  \\
    12       & 32,302    & 15.77 & 864    & 32,282    & 15.76 & 171 & 32,761    & 15.99 & 34  \\
    13       & 64,917    & 15.85 & 1,280  & 64,884    & 15.84 & 160 & 65,531    & 15.99 & 35  \\
    14       & 130,214   & 15.90 & 1,700  & 130,169   & 15.90 & 199 & 131,067   & 15.99 & 35  \\
    15       & 260,905   & 15.92 & 2,502  & 260,886   & 15.92 & 261 & 262,133   & 15.99 & 35  \\
    16       & 522,434   & 15.94 & 3,471  & 522,451   & 15.94 & 276 & 524,273   & 15.99 & 34  \\
    17       & 1,046,118 & 15.96 & 5,085  & 1,045,962 & 15.96 & 416 & 1,048,561 & 15.99 & 36  \\
    18       & 2,093,484 & 15.97 & 7,029  & 2,093,526 & 15.97 & 465 & 2,097,142 & 15.99 & 41  \\
    19       & 4,189,181 & 15.98 & 10,222 & 4,189,155 & 15.98 & 606 & 4,194,296 & 15.99 & 38  \\
    20       & 8,381,379 & 15.99 & 14,374 & 8,381,431 & 15.99 & 644 & 8,388,594 & 15.99 & 37  \\ \bottomrule
    \end{tabular}
    \caption{Number of edges $m$, average and maximum degrees ($\Delta$) for the R-MAT graphs generated for each scale $x \in [10..20]$ and initiator matrix.}
    \label{tab:rmat_graph_stats}
\end{table*}